\documentclass{llncs}

\usepackage{graphicx}
\usepackage{pgf}
\usepackage{tikz}
\usepackage{amssymb}
\usepackage{amsmath}
\usetikzlibrary{arrows,automata,shapes,decorations}
\usetikzlibrary{decorations.pathmorphing}
\usetikzlibrary{decorations.markings}
\usepackage{verbatim} 
\usepackage{wrapfig}
\usepackage{color}

\usepackage{algorithmicx}
\usepackage{algpseudocode}
\usepackage{array}

\newcommand{\attr}{\textup{Attr}}
\newcommand{\infin}{\textup{inf}}

\newcommand{\gamea}{\mathcal{G}}

\newcommand{\pa}{\textup{0}}
\newcommand{\pb}{\textup{1}}

\newcommand{\play}{\textup{Play}}
\newcommand{\hist}{\textup{Hist}}
\newcommand{\path}{\textup{Path}}
\newcommand{\st}{\textup{st}}
\newcommand{\val}{\textup{val}}
\newcommand{\dist}{\textup{dist}}
\newcommand{\nextstate}{\textup{next}}

\newcommand{\nats}{\mathbb{N}}

\newcommand{\np}{\textsc{NP}}
\newcommand{\conp}{\textsc{co-NP}}
\newcommand{\up}{\textsc{UP}}
\newcommand{\coup}{\textsc{co-UP}}

\newcommand{\mycut}[1]{}

\makeatletter
\newsavebox{\@brx}
\newcommand{\llangle}[1][]{\savebox{\@brx}{\(\m@th{#1\langle}\)}%
  \mathopen{\copy\@brx\kern-0.5\wd\@brx\usebox{\@brx}}}
\newcommand{\rrangle}[1][]{\savebox{\@brx}{\(\m@th{#1\rangle}\)}%
  \mathclose{\copy\@brx\kern-0.5\wd\@brx\usebox{\@brx}}}

\begin{document}

\title{Winning Cores in Parity Games}

\author{Steen Vester}

\institute{Technical University of Denmark, Kgs. Lyngby, Denmark\\ \email{stve@dtu.dk}}

\maketitle

\begin{abstract}

We introduce the novel notion of winning cores in parity games and develop a deterministic polynomial-time under-approximation algorithm for solving parity games based on winning core approximation. Underlying this algorithm are a number properties about winning cores which are interesting in their own right. In particular, we show that the winning core and the winning region for a player in a parity game are equivalently empty. Moreover, the winning core contains all fatal attractors but is not necessarily a dominion itself. Experimental results are very positive both with respect to quality of approximation and running time. It outperforms existing state-of-the-art algorithms significantly on most benchmarks.
\end{abstract}

\section{Introduction}

Solving parity games \cite{EJ91} is an important problem of both theoretical and practical interest. It is known to be in $\np \cap \conp$ \cite{EJS01} and $\up \cap \coup$ \cite{Jur98} but in spite of the development of many different algorithms (see e.g. \cite{Zie98,Jur00,VJ00,JPZ06,Sch07}), frameworks for benchmarking such algorithms \cite{FL09,Kei14} and families of parity games designed to expose the worst-case behaviour of existing algorithms \cite{Jur00,Fri09,Fri11} it has remained an open problem whether a polynomial-time algorithm exists.

Various problems for which polynomial-time algorithms are not known can been reduced in polynomial time to the problem of solving parity games. Among these are model-checking of the propositional $\mu$-calculus \cite{Koz83,EL86,Sti95}, the emptiness problem for parity automata on infinite binary trees \cite{Mos84,EJS01} and solving boolean equation systems \cite{Mal97}. For relations to other problems in logic and automata theory, see e.g. \cite{GTW02}.

Some of the most notable algorithms from the litterature of solving parity games include Zielonkas algorithm \cite{Zie98} using $O(n^d)$ time, the small progress measures algorithm \cite{Jur00} using $O(d \cdot m \cdot (n/d)^{d/2})$ time, the strategy improvement algorithm \cite{VJ00} using $O(n\cdot m \cdot 2^m)$ time, the big step algorithm \cite{Sch07} using $O(m\cdot n^{d/3})$ time and the dominion decomposition algorithm \cite{JPZ06} using $O(n^{\sqrt{n}})$ time. Here, $n$ is the number of states in the game, $m$ is the number of transitions and $d$ is the maximal color of the game.

The main contributions of this paper are to introduce the novel concept of \emph{winning cores} in parity games and develop a fast deterministic polynomial-time under-approximation algorithm for solving parity games based on properties of winning cores. Two different, but equivalent, definitions of winning cores are given both of which are used to show a number of interesting properties. One is based on the new notion of \emph{consecutive dominating sequences}.

We perform an investigation of winning cores and show that the winning core of a player is always a subset of the winning region of the player and more importantly that the winning core of a player is empty if and only if the winning region of the player is empty. A result of \cite{DKT12} then implies that emptiness of the winning core of a player can be decided in polynomial time if and only if parity games can be solved in polynomial time. We further show that the winning cores for the two players contain all fatal attractors \cite{HKP13,HKP14} and show some recursive properties of winning cores which are similar in nature to the properties of winning regions that form the basis of the recursive algorithms in \cite{Zie98,JPZ06,Sch07} for solving parity games.

We further show that winning cores are not necessarily dominions \cite{JPZ06} which is interesting on its own. To the knowledge of the author no meaningful subsets of the winning regions have been characterized in the litterature which were not dominions. As such, several of the existing algorithms for solving parity games are based on finding dominions, e.g. \cite{Zie98,JPZ06,Sch07}. However, it was recently shown in \cite{GLMMO15} that there is no algorithm which decides if there exists a dominion with at most $k$ states in time $n^{o(\sqrt{k})}$ unless the exponential-time hypothesis fails. Thus, going beyond dominions could very well be important in the search for a polynomial-time algorithm for solving parity games. Winning cores provide a viable direction for this search.

Next, we show the existence of memoryless optimal strategies for games with a certain type of prefix-dependent objectives using a result of \cite{GZ05}. Based on this we provide a decreasing sequence of sets of states which converges to the winning core in at most $n$ steps. It is also shown that winning cores can be computed in polynomial time if and only if parity games can be solved in polynomial time and that winning core computation is in $\up \cap \coup$ by a reduction to solving parity games.

The correctness of the under-approximation algorithm relies on fast convergence of the sequence mentioned above. It uses $O(d \cdot n^2 \cdot (n+m))$ time and $O(d+n+m)$ space. It is an under-approximation algorithm in the sense that it returns subsets of the winning regions for the two players.

The algorithm has been implemented in OCaml on top of the \textsc{PgSolver} framework \cite{FL09} and experiments have been carried out both to test the quality of the approximations as well as the practical running times. The experimental results are very positive as it solved all games from the benchmark set of \textsc{PgSolver} completely and solved a very high ratio of randomly generated games completely. Further, on most of the benchmark games it outperformed the existing state-of-the-art algorithms significantly and solved games with more than $10^7$ states. The algorithm also performed very well compared to the best existing partial solver for parity games \cite{HKP13,HKP14} both with respect to quality of approximation and running time.

Chapter 2 contains preliminary definitions and Chapter 3 introduces consecutive dominating sequences. In Chapter 4 winning cores are introduced and a number of properties about them are presented. In Chaper 5 the computational complexity of computing winning cores is analyzed. In Chapter 6 the approximation algorithm is presented and Chapter 7 contains experimental results. Finally, Chapter 8 contains concluding remarks.

\section{Preliminaries}

A \emph{parity game} \cite{EJ91} is played by two players called player $\pa$ and player $\pb$. It is played in a finite transition system where the states are partitioned into states that player $\pa$ controls and states that player $\pb$ controls. Further, each state is colored with a natural number. The game is played by placing a token in an initial state $s_0$. The player controlling the current state must choose a successor state to move the token to while respecting the transition relation. Then the player controlling the successor state chooses a new successor state and so on indefinitely. We require that the transition relation is total and thus the play is always an infinite sequence of states. Player $\pa$ wins if the greatest color occuring infinitely often along the play is even and player $\pb$ wins if the greatest color occuring infinitely often is odd.

\subsection{Basic definitions}

More formally we define parity games as follows.

\begin{definition}
 A parity game is a tuple $\gamea = (S, S_\pa, S_\pb, R, c)$ such that
 \begin{itemize}
  \item $S$ is a finite set of states
  \item $S_\pa$ and $S_\pb$ partitions $S$. That is, $S_\pa \cup S_\pb = S$ and $S_\pa \cap S_\pb = \emptyset$  
  \item $R \subseteq S \times S$ is the transition relation which is total  
  \item $c: S \mapsto \{1,...,d\}$ is a coloring function specifying a color for each state  
 \end{itemize}
\end{definition}

\begin{example}
\label{ex:game}
 A simple example of a parity game can be seen in Figure \ref{fig:first_game}. Circle states are in $S_0$ and square states in $S_1$. The values drawn inside states are colors. There is an arrow from state $s$ to state $t$ if $(s,t) \in R$.
\begin{figure}
 \begin{center}
 \begin{tikzpicture}

\tikzstyle{every node}=[ellipse, draw=black,
                        inner sep=0pt, minimum width=9pt, minimum height=9pt]

\draw (2,7) node [rectangle, label = below right: $s_0$] (s0) {$2$};    
\draw (3,7) node [label = below right: $s_1$] (s1) {$1$};
\draw (1,7) node [label = below right: $s_2$] (s2) {$5$};
\draw (0,7) node [label = below right: $s_3$] (s3) {$4$};
\draw (4,7) node [label = below right: $s_4$] (s4) {$3$};

\path[->] (s0) edge [bend left] (s1) {};
\path[->] (s1) edge [bend left] (s0) {};
\path[->] (s0) edge (s2) {};
\path[->] (s2) edge (s3) {};
\path[->] (s1) edge (s4) {};
\path[->] (s3) edge [loop left] (s3) {};
\path[->] (s4) edge [loop right] (s4) {};

\end{tikzpicture}
\end{center}
\caption{Example of a parity game.}
\label{fig:first_game}
\end{figure}
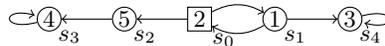
 \end{example}

 For the rest of this section and Section \ref{sec:dom} and \ref{sec:core} we fix a parity game $\gamea = (S, S_0, S_1, R, c)$ with colors in $\{1,...,d\}$ and $n$ states. 
 
 For a set $X$ we denote by $X^*, X^+$ and $X^\omega$ the sets of finite, non-empty finite and infinite sequences of elements in $X$ respectively. For a sequence $\rho = s_0 s_1 ...$ we let $\rho_{\geq i} = s_i s_{i+1} ...$, $\rho_{> i} = s_{i+1} s_{i+2} ...$, $\rho_{\leq i} = s_0 s_1 ... s_i$ and $\rho_i = s_i$. A \emph{play} is an infinite sequence $\rho = s_0 s_1 ... \in S^\omega$ that respects the transition relation. That is, $(s_i, s_{i+1}) \in R$ for all $i \geq 0$. A \emph{history} is a finite, non-empty prefix of a play. A \emph{path} is either a history or a play. The set of plays,  histories and paths in $\gamea$ are denoted $\play(\gamea)$, $\hist(\gamea)$ and $\path(\gamea)$ respectively. The set of plays, histories and paths with initial state $s_0$ are denoted $\play(\gamea, s_0)$, $\hist(\gamea,s_0)$ and $\path(\gamea, s_0)$ respectively.
 
 For a path $\rho = s_0 s_1 ...$ in $\gamea$ we define $c(\rho) = c(s_0) c(s_1) ...$ and a set $P$ of paths define $c(P) = \{c(\rho) \mid \rho \in P\}$. For a sequence $\pi = e_0 e_1 ...$ let $\infin(\pi) = \{e \mid \textup{there exists infinitely many } i \textup{ s.t. } e = e_i \}$. We define the \emph{parity objective} $\Omega_j$ for player $j \in \{\pa, \pb \}$ by
$$\Omega_j = \{\pi \in \nats^\omega \mid \exists k . \forall i \geq 0 . \pi_i \leq k \wedge \max (\infin(\pi)) \equiv j \textup{ (mod } 2) \}$$
where $\nats$ is the set of non-negative integers. Further, let
$$\Lambda_j = \{\pi \in \Omega_j \mid \max_{i > 0} \pi_i \equiv j \textup{ (mod } 2) \}$$
Note that the initial element of the sequence is not counted in the definition of $\Lambda_j$.

A \emph{strategy} for player $j$ is a partial function $\sigma_j: \hist(\gamea) \rightarrow S$ defined for histories $s_0 ... s_k$ such that $s_k \in S_j$ with the requirement that $(s_k, \sigma_j(s_0 ... s_k)) \in R$. A \emph{memoryless strategy} for player $j$ is a strategy $\sigma$ such that $\sigma(h) = \sigma(h')$ whenever the last state of $h$ and the last state of $h'$ are the same. With slight abuse of notation we write $\sigma(s) = s'$ when the memoryless strategy $\sigma$ chooses successor state $s'$ for all histories $h$ ending in $s$.

A play (resp. history) $s_0 s_1 ...$ (resp. $s_0 ... s_k$) is \emph{compatible} with strategy $\sigma_j$ for player $j$ if $\sigma_j(s_0 ... s_i) = s_{i+1}$ for each $i \geq 0$ (resp. $0 \leq i < k$) such that $s_i \in S_j$. The set of plays (resp. histories) compatible with strategy $\sigma_j$ is denoted $\play(\gamea, \sigma_j)$ (resp. $\hist(\gamea, \sigma_j)$). The subsets where we restrict to plays, histories and paths with initial state $s_0$ are denoted $\play(\gamea, s_0, \sigma_j)$, $\hist(\gamea, s_0, \sigma_j)$ and $\path(\gamea, s_0, \sigma_j)$ respectively.

We say that $\sigma_j$ is a \emph{winning strategy} for player $j$ from state $s_0$ if $c(\play(\gamea, s_0, \sigma_j)) \subseteq \Omega_j$. When such a strategy exists we call $s_0$ a \emph{winning state} for player $j$. We write $W_j(\gamea)$ for the set of winning states of player $j$ in $\gamea$. Since parity games are memoryless determined \cite{EM79,EJ91} we have $W_0(\gamea) \cup W_1(\gamea) = S$ and $W_0(\gamea) \cap W_1(\gamea) = \emptyset$. Further, 
there is a memoryless strategy for player $j$ that is winning from every $s \in W_j(\gamea)$.

\subsection{Restricted parity games}

We define the \emph{restricted parity game} $\gamea \restriction S' = (S', S'_\pa, S'_\pb, R', c')$ for a subset $S' \subseteq S$ by
\begin{itemize}
 \item $S'_j = S' \cap S_j$ for $j \in \{\pa, \pb\}$
  
 \item $R' = R \cap (S' \times S')$
 
 \item $c'(s) = c(s)$ for every $s \in S'$
 
\end{itemize}

Intuitively, the restricted parity game $\gamea \restriction S'$ is the same as $\gamea$ where all states not in $S'$ are removed and all transitions $(s,s')$ with either $s$ or $s'$ not in $S'$ are removed. Note that the restricted parity game is only a well-defined parity game when $R'$ is total.

\subsection{Attractor sets}

The notion of an \emph{attractor set} is well-known \cite{Zie98} and is the set of states from which a player $j$ can ensure reaching a set of target states.

\begin{definition}
 The attractor set $\attr_j(\gamea, T)$ for a target set $T \subseteq S$ and a player $j$ is the limit of the sequence $\attr^i_j(\gamea, T)$ where
\\
\\
\begin{tabular}{>{\hspace{-0.5pc}}l l<{\hspace{-2pt}}}
$\attr^0_j(\gamea, T) $ & $= T$ \\
$\attr^{i+1}_j(\gamea, T) $ & $= \attr^i_j(\gamea, T)$\\
& $\cup \{s \in S_j \mid \exists t . (s,t) \in R \wedge t \in \attr^i_j(\gamea, T) \}$ \\
& $\cup \{s \in S_{1-j} \mid \forall t . (s,t) \in R \Rightarrow t \in \attr^i_j(\gamea, T) \}$\\
\end{tabular}

\end{definition}

The attractor set and a memoryless strategy to ensure reaching the target from this set can be computed in time $O(n+m)$ in a game with $n$ states and $m$ transitions \cite{dAHK98}. The positive attractor $\attr^+_j(\gamea, T)$ is the set of states from which player $j$ can ensure reaching $T$ in at least 1 step. Formally,
\\
\\
\begin{tabular}{l l}
 $\attr^+_j(\gamea,T) = $ & $\attr_j(\gamea,$ \\
 & $\{s \in S_j \mid \exists t \in T . (s,t) \in R\}$ \\
 & $\cup \{s \in S_{1-j} \mid \forall t \in S. (s,t) \in R \Rightarrow t \in T\})$
\end{tabular}

\subsection{$j$-closed sets and dominions}

A subset $S' \subseteq S$ of states in a parity game is called $j$-closed if

\begin{enumerate}
 \item For every $s \in S_{1-j} \cap S'$ there exists no $t \in S \setminus S'$ such that $(s,t) \in R$
 
 \item For every $s \in S_j \cap S'$ there exists $t \in S'$ such that $(s,t) \in R$
\end{enumerate}

Thus, a set of states is $j$-closed if and only if player $j$ can force the play to stay in this set of states.

A $j$-dominion \cite{JPZ06} for player $j$ is a set $T \subseteq S$ such that from every state $s \in T$ player $j$ has a strategy $\sigma$ such that $c(\play(\gamea, s, \sigma)) \subseteq \Omega_j$ and $\play(\gamea, s, \sigma) \subseteq T^\omega$. That is, from every state in a $j$-dominion, player $j$ can ensure to win while keeping the play inside the $j$-dominion. Thus, a $j$-dominion is $j$-closed.

\begin{proposition}[\cite{JPZ06}]
 $W_j(\gamea)$ is a $j$-dominion.
\end{proposition}

\begin{proposition}[\cite{JPZ06}]
\label{prop:remove_winning}
 Let $V \subseteq W_j(\gamea)$, $V' = \attr_j(\gamea, V)$ and $\gamea' = \gamea \restriction (S \setminus V')$. Then $W_j(\gamea) = V' \cup W_j(\gamea')$ and $W_{1-j}(\gamea) = W_{1-j}(\gamea')$.
\end{proposition}

Many of the existing algorithms for solving parity games work by finding a dominion $D$ for some player $j$ and then apply Proposition \ref{prop:remove_winning} to remove the states in $\attr_j(\gamea, D)$ and recursively solve the smaller resulting game. This includes Zielonkas algorithm \cite{Zie98}, the dominion decomposition algorithm \cite{JPZ06} and the big step algorithm \cite{Sch07}. The algorithm we present in this paper also applies this proposition, but the winning cores which we search for are not necessarily dominions.

\section{Dominating Sequences}
\label{sec:dom}
We say that a path $\rho = s_0 s_1 ...$ with at least one transition is \emph{0-dominating} if the color $e = \max \{c(s_i) \mid i > 0 \}$ is even and \emph{1-dominating} if it is odd. Note that we do not include the color of the first state of the sequence.

\begin{figure}
\begin{center}
\begin{tikzpicture}

\tikzstyle{every node}=[ellipse, draw=black,
                        inner sep=0pt, minimum width=9pt, minimum height=9pt]

\draw (0,7) node [label=below left: $s_0$] (s0) {$1$};    
\draw (0.6,7) node [label=below left: $s_1$] (s1) {$4$};
\draw (1.2,7) node [label=below left: $s_2$] (s2) {$3$};
\draw (1.8,7) node [label=below left: $s_3$] (s3) {$4$};
\draw (2.4,7) node [label=below left: $s_4$] (s4) {$3$};

\draw (3.6,7) node [label=below left: $t_0$] (t0) {$6$};    
\draw (4.2,7) node [label=below left: $t_1$] (t1) {$2$};
\draw (4.8,7) node [label=below left: $t_2$] (t2) {$3$};
\draw (5.4,7) node [label=below left: $t_3$] (t3) {$2$};
\draw (6.0,7) node [label=below left: $t_4$] (t4) {$3$};
\draw (6.6,7) node [label=below left: $t_5$] (t5) {$2$};
\draw (7.2,7) node [label=below left: $t_6$] (t6) {$3$};
\draw (7.8,7) node [draw=none] (t7) {$...$};

\path[->] (s0) edge (s1) {};
\path[->] (s1) edge (s2) {};
\path[->] (s2) edge (s3) {};
\path[->] (s3) edge (s4) {};

\path[->] (t0) edge (t1) {};
\path[->] (t1) edge (t2) {};
\path[->] (t2) edge (t3) {};
\path[->] (t3) edge (t4) {};
\path[->] (t4) edge (t5) {};
\path[->] (t5) edge (t6) {};
\path[->] (t6) edge (t7) {};

\path[-] (-0.3,6.5) edge [line width = 2pt] (0.9,6.5) {};
\path[-] (0.3,6.3) edge [line width = 2pt] (2.1,6.3) {};

\path[-] (3.3,6.5) edge [line width = 2pt] (5.1,6.5) {};
\path[-] (4.5,6.3) edge [line width = 2pt] (6.3,6.3) {};
\path[-] (5.7,6.1) edge [line width = 2pt] (7.5,6.1) {};
\end{tikzpicture}
\end{center}
\caption{Consecutive $j$-dominating sequences illustrated by bold lines. Note the overlap of one state between sequences.}
\label{fig:dominating}
\end{figure}
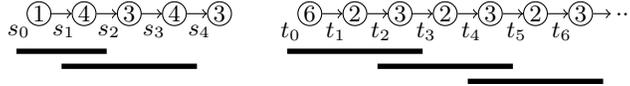

We say that a path $\rho$ begins with $k$ consecutive $j$-dominating sequences if there exists indices $i_0 < i_1 < ... < i_k$ with $i_0 = 0$ such that $\rho_{i_\ell} \rho_{i_\ell + 1} ... \rho_{i_{\ell+1}}$ is $j$-dominating for all $0 \leq \ell < k$. Similarly, a play $\rho$ begins with an infinite number of consecutive $j$-dominating sequences if there exists an infinite sequence $i_0 < i_1 < ...$ of indices with $i_0 = 0$ such that $\rho_{i_\ell} \rho_{i_\ell + 1} ... \rho_{i_{\ell + 1}}$ is $j$-dominating for all $\ell \geq 0$.

As examples, the sequence on the left in Figure \ref{fig:dominating} begins with two consecutive 0-dominating sequences $s_0 s_1$ and $s_1 s_2 s_3$ whereas the sequence to the right begins with only one 0-dominating sequence $t_0 t_1$, but not two consecutive 0-dominating sequences. Also, the sequence to the left does not begin with a 1-dominating sequence whereas the sequence to the right begins with an infinite number of consecutive 1-dominating sequences: $t_0 t_1 t_2$, $t_2 t_3 t_4$, $t_4 t_5 t_6$ etc.

We start with the following well-known lemma, stating that the winner of a play $\rho$ in a parity game is independent of a given finite prefix of the play.

\begin{lemma}
\label{lem:suffix_winning}
 Let $\rho$ be a play and $\rho'$ be a suffix of $\rho$. Then $c(\rho) \in \Omega_j$ if and only if $c(\rho') \in \Omega_j$.
\end{lemma}

The following proposition shows that a play is winning for player $j$ if and only if it has a suffix that begins with an infinite number of consecutive $j$-dominating sequences.

\begin{proposition}
\label{prop:winning_iff_jdoms}
 Let $\rho$ be a play. Then $c(\rho) \in \Omega_j$ if and only if there is a suffix of $\rho$ that begins with an infinite number of consecutive $j$-dominating sequences.
\end{proposition}

Next, we show a slightly surprising fact. A play begins with an infinite number of consecutive $j$-dominating sequences if and only if it is both winning for player $j$ and $j$-dominating. This means that we have two quite different definitions of the same concept. Both will be used to obtain results later in the paper.

\begin{proposition}
\label{prop:cons_iff_jdom}
 Let $\rho$ be a play. Then $\rho$ begins with an infinite number of consecutive $j$-dominating sequences if and only if $\rho$ is $j$-dominating and $c(\rho) \in \Omega_j$.
\end{proposition}

\begin{proof}
 $(\Leftarrow)$ Suppose that $\rho$ is $j$-dominating and $c(\rho) \in \Omega_j$. By Proposition \ref{prop:winning_iff_jdoms} there exists a suffix $\rho_{\geq \ell}$ of $\rho$ that begins with an infinite number of consecutive $j$-dominating sequences $\rho^0, \rho^1, ...$. Let $i_0 < i_1 < ...$ be the indices such that $\rho^\ell = \rho_{i_\ell} \rho_{i_\ell + 1} ... \rho_{i_{\ell + 1}}$. Since $\rho$ is $j$-dominating the greatest color $e$ of a non-initial state in $\rho$ satisfies $e \equiv j \textup{ (mod } 2)$. Let $k \geq 0$ be the smallest index such that $\max\{c(\rho_i) \mid 0 < i \leq i_{k+1}\} = e$. Now we have that $\rho$ begins with an infinite number of consecutive $j$-dominating sequences, namely $\rho_{\leq i_{k+1}}, \rho^{k+1}, \rho^{k+2} ...$.
 
 $(\Rightarrow)$ By Proposition \ref{prop:winning_iff_jdoms} we also have that if $\rho$ begins with an infinite number of consecutive $j$-dominating sequences $\rho^0, \rho^1, ...$ then $c(\rho) \in \Omega$. Now, suppose for contradiction that $\rho$ is not $j$-dominating. Then the largest color $e$ of a non-initial state in $\rho$ satisfies $e \equiv 1-j \textup{ (mod } 2)$. Let $i > 0$ be an index such that $c(\rho_i) = e$. Let $k$ be an index such that $\rho^k$ contains $\rho_i$ as a non-initial state. As $\rho^k$ is $j$-dominating there is a non-initial state in $\rho^k$ with a color $e' > e$ which gives a contradiction.
\end{proof}

As the two characterizations are equivalent, we will just write $c(\rho) \in \Lambda_j$ for the remainder of the paper when we know that either property is true of $\rho$.


\section{Winning cores}
\label{sec:core}
We define \emph{the winning core} $A_j(\gamea)$ for player $j$ in the parity game $\gamea$ as the set of states $s$ from which player $j$ has a strategy $\sigma$ such that $c(\play(\gamea, s, \sigma)) \subseteq \Lambda_j$. According to Proposition \ref{prop:cons_iff_jdom} we have two different characterizations of this set of plays. Both will be used in the following depending on the application.

\subsection{The winning core and the winning region}

First note that since $\Lambda_j \subseteq \Omega_j$ we have that every state in the winning core for player $j$ is a winning state for player $j$.

\begin{proposition}
\label{prop:core_subset}
 Let $\gamea$ be a parity game. Then $A_j(\gamea) \subseteq W_j(\gamea)$.
\end{proposition}


Next, we will show a more surprising fact: If the winning core for player $j$ is empty, then the winning region of player $j$ is empty as well. This is a very important property of winning cores.
 
\begin{proposition}
\label{prop:empty_core}
 Let $\gamea$ be a parity game. If $A_j(\gamea) = \emptyset$ then $W_j(\gamea) = \emptyset$.
\end{proposition}

\begin{proof}
 Let $A_j(\gamea) = \emptyset$. Suppose for contradiction that $W_j(\gamea) \neq \emptyset$. Then there exists $s \in W_j(\gamea)$ and a memoryless winning strategy $\sigma$ for player $j$ from $s$.
 
 Since $s \not \in A_j(\gamea)$ there exists $\rho \in \play(\gamea, s, \sigma)$ that does not begin with an infinite number of consecutive $j$-dominating sequences. However, since $c(\rho) \in \Omega_j$ there is a suffix of $\rho$ that begins with an infinite number of consecutive $j$-dominating sequences. Let $\ell_0 > 0$ be the smallest index such that $\rho_{\geq \ell_0}$ begins with an infinite number of consecutive $j$-dominating sequences. Then $\rho_{\leq \ell_0}$ is $(1-j)$-dominating, because otherwise $\rho$ would begin with an infinite number of consecutive $j$-dominating sequences.
 
 Since $A_j(\gamea) = \emptyset$ there exists $\rho' \in \play(\gamea, \rho_{\ell_0}, \sigma)$ that does not begin with an infinite number of consecutive $j$-dominating sequences. Since $\sigma$ is memoryless we have $\rho_{< \ell_0} \cdot \rho' \in \play(\gamea, s, \sigma)$ which means that $c(\rho') \in \Omega_j$ according to Lemma \ref{lem:suffix_winning}. Here, $\cdot$ is the concatenation operator. This implies that there is a suffix of $\rho'$ that begins with an infinite number of consecutive $j$-dominating sequences. Let $\ell_1 > 0$ be the smallest index such that $\rho'_{\geq \ell_1}$ begins with an infinite number of consecutive $j$-dominating sequences. Then $\rho'_{\leq \ell_1}$ is $(1-j)$-dominating, because otherwise $\rho'$ would begin with an infinite number of consecutive $j$-dominating sequences.
 
 Since $A_j(\gamea) = \emptyset$ there exists $\rho'' \in \play(\gamea, \rho'_{\ell_1},\sigma)$ that does not begin with an infinite number of consecutive $j$-dominating sequences. Since $\sigma$ is memoryless we have that $\rho_{< \ell_0} \cdot \rho'_{< \ell_1} \cdot \rho'' \in \play(\gamea, s, \sigma)$ which means that $c(\rho'') \in \Omega_j$ according to Lemma \ref{lem:suffix_winning}. We can continue this construction in the same way to obtain the play $\pi = \rho_{< \ell_0} \cdot \rho'_{< \ell_1} \cdot \rho''_{< \ell_2} \cdot ... $ which belongs to $\play(\gamea, s, \sigma)$. The construction is illustrated in Figure \ref{fig:inf_contradict}.
 
 \begin{figure}
\begin{center}
\begin{tikzpicture}

\tikzstyle{every node}=[ellipse, draw=black,
                        inner sep=0pt, minimum width=3pt, minimum height=3pt]

\draw (5,8) node [label=above: $s$] (s0) {};    
\draw (4.5,7) node [label=right: $\rho_{\ell_0}$] (s1) {};
\draw (4.8,6) node [label=right: $\rho'_{\ell_1}$] (s2) {};
\draw (5.6,5.4) node [label=above right: $\rho''_{\ell_2}$] (s3) {};

\draw (3.5,4) node [draw=none, minimum height=10pt] (end0) {...};
\draw (3.5,6) node [draw=none] (label0) {$\rho$};

\draw (4.5,4) node [draw=none, minimum height=10pt] (end1) {...};
\draw (4.5,5.5) node [draw=none] (label1) {$\rho'$};

\draw (5.5,4) node [draw=none, minimum height=10pt] (end2) {...};
\draw (5.3,4.7) node [draw=none] (label2) {$\rho''$};

\draw (6.6,5.2) node [draw=none, minimum height=10pt] (end3) {...};

\draw (s0) .. controls (4.7,7.8) and (4.7,7.4) .. (s1);
\draw (s1) .. controls (4.6,6.8) and (4.5,6.5) .. (s2);
\draw (s2) .. controls (5,5.6) .. (s3);
\draw (s3) .. controls (5.8,5.3) .. (end3);

\draw [dashed] (s1) .. controls (4,6.5) and (4,4.5) .. (end0) {};
\draw [dashed] (s2) .. controls (5,5.5) and (5,4.5) .. (end1) {};
\draw [dashed] (s3) .. controls (5.8,5.1)  .. (end2) {};

\end{tikzpicture}
\end{center}
\caption{Construction of a play $\pi \in \play(\gamea, s, \sigma)$ that begins with an infinite number of consecutive $(1-j)$-dominating sequences. $\pi$ is the solid path in the figure.}
\label{fig:inf_contradict}
\end{figure}
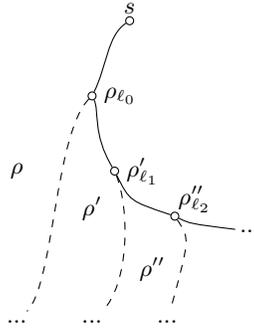

 Observe that $\pi$ begins with an infinite number of consecutive $(1-j)$-dominating sequences, namely $\rho_{\leq \ell_0}, \rho'_{\leq \ell_1}, \rho''_{\leq \ell_2}$,... which are all $(1-j)$-dominating. By Proposition \ref{prop:winning_iff_jdoms} we have $c(\pi) \in \Omega_{1-j}$. This is a contradiction since $\pi \in \play(\gamea, s, \sigma)$ and $c(\pi) \in \Omega_j$. Thus, $W_j(\gamea) = \emptyset$.
 \end{proof}

Proposition \ref{prop:core_subset} and \ref{prop:empty_core} give us the following result.

\begin{theorem}
\label{theo:empty_iff}
  Let $\gamea$ be a parity game. The winning core $A_j(\gamea)$ for player $j$ in $\gamea$ is empty if and only if the winning region $W_j(\gamea)$ for player $j$ in $\gamea$ is empty.
\end{theorem}

\begin{remark}
 As shown in \cite{DKT12} parity games can be solved in polynomial time if and only if it can be decided in polynomial time whether the winning region $W_j(\gamea) = \emptyset$ for player $j$. Thus, Theorem \ref{theo:empty_iff} implies that parity games can be solved in polynomial time if and only if emptiness of the winning core for player $j$ can be decided in polynomial time.
\end{remark}

In \cite{HKP13} the concept of a \emph{fatal attractor} is defined and used for partially solving parity games. A fatal attractor is a set $X$ of states colored $e \equiv j \textup{ (mod 2)}$ with the property that player $j$ can ensure that when the play begins in a state in $X$ then it will eventually reach $X$ again without having passed through any states with color greater then $e$ along the way. 
Player $j$ can thus force the play to begin with an infinite number of consecutive $j$-dominating sequences from states in $X$ by repeatedly forcing the play back to $X$ in this fashion.

\begin{proposition}
 Let $X$ be a fatal attractor for player $j$ in $\gamea$. Then $X \subseteq A_j(\gamea)$.
\end{proposition}

Note that the winning core $A_j(\gamea)$ for player $j$ need not be a $j$-dominion. In addition, neither does $\attr_j(\gamea, A_j(\gamea))$. Indeed, consider the parity game in Figure \ref{fig:not_dominion}. In this game, the winning core for player 0 is $A_0(\gamea) = \{s_0,s_3\}$ and also $\attr_0(\gamea, A_0(\gamea)) = A_0(\gamea)$. Clearly, this is not a 0-dominion as player 1 can force the play to go outside this set. Note also that this game has no fatal attractors. Thus, the winning core can contain more states than just the fatal attractors.

The property that the winning core for player $j$ is not necessarily a $j$-dominion is interesting as, to the knowledge of the author, no meaningful subsets of the winning region for player $j$ that are not necessarily $j$-dominions have been characterized in the litterature. Thus, many algorithms focus on looking for dominions which can be removed from the game, e.g. the algorithms from \cite{Zie98,JPZ06,Sch07}. However, it was recently shown in \cite{GLMMO15} that there is no algorithm which decides if there exists a dominion with at most $k$ states in time $n^{o(\sqrt{k})}$ unless the exponential-time hypothesis fails. This, along with Theorem \ref{theo:empty_iff} make winning cores very interesting objects for further study as they propose a fresh direction of research in solving parity games.

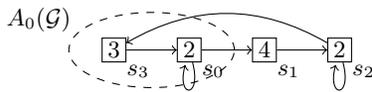
\begin{figure}
\begin{center}
\begin{tikzpicture}

\tikzstyle{every node}=[ellipse, draw=black,
                        inner sep=0pt, minimum width=9pt, minimum height=9pt]


\draw (3,3) node [rectangle, label = below right: $s_0$] (t0) {$2$};    
\draw (4,3) node [rectangle, label = below right: $s_1$] (t1) {$4$};
\draw (5,3) node [rectangle, label = below right: $s_2$] (t2) {$2$};
\draw (2,3) node [rectangle, label = below right: $s_3$] (t3) {$3$};
\draw (1,3.4) node [draw=none] (s4) {$A_0(\gamea)$};



\path[->] (t0) edge (t1) {};
\path[->] (t1) edge (t2) {};
\path[->] (t2) edge [bend right] (t3) {};
\path[->] (t3) edge (t0) {};
\path[->] (t0) edge [loop below] (t0) {};
\path[->] (t2) edge [loop below] (t2) {};

\draw (2.5,3) ellipse (1.1cm and 0.5cm) [dashed];

\end{tikzpicture}
\end{center}

\caption{A parity game where the winning core for player 0 is not a $0$-dominion}
\label{fig:not_dominion}
\end{figure}

In the remainder of this subsection we state some more interesting properties about winning cores which are similar in nature to the results on winning regions in parity games from \cite{Zie98} that form the basis of Zielonkas algorithm for solving parity games as well as optimized versions in \cite{JPZ06,Sch07}.

Let $\gamea = (S,S_0,S_1,R,c)$ be a parity game with largest color $d$. Let $k$ be the player such that $d \equiv k \textup{ (mod 2)}$. Let $S^d$ be the set of states with color $d$ and $U = \attr^+_{k}(\gamea, S^d)$. Let $\gamea' = \gamea \restriction (S \setminus U)$. We define $\gamea'' = \gamea \restriction (S \setminus \attr_{1-k}(\gamea, A_{1-k}(\gamea))$ as the parity game obtained from $\gamea$ by removing the set of states from which player $1-k$ can force the play to go to his winning core in $\gamea$.

\begin{proposition}
\label{prop:core_1_k}
 $A_{1-k}(\gamea) = A_{1-k}(\gamea')$
\end{proposition}

\begin{proposition}
\label{prop:core_k}
$A_{k}(\gamea) = A_{k}(\gamea'')$
\end{proposition}

The situation is illustrated in Figure \ref{fig:recurse} where the winning regions and winning cores for the two players in $\gamea'$ are shown as well. Note that $A_{1-k}(\gamea)$ is contained in $W_{1-k}(\gamea')$ but that $A_{k}(\gamea)$ can contain states in $U$.
 
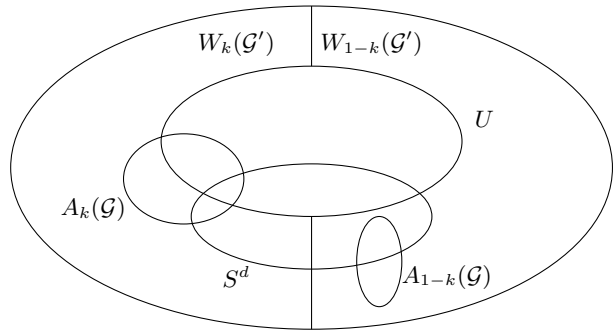
\begin{figure}
\begin{center}
\begin{tikzpicture}

\tikzstyle{every node}=[ellipse, draw=black,
                        inner sep=0pt, minimum width=3pt, minimum height=3pt]

\draw (6,6.35) ellipse (4cm and 2.15cm);                    
\draw (6,6.7) ellipse (2cm and 1cm);


\draw (6,5.7) ellipse (1.6cm and 0.7cm);

\draw (4.3,6.2) ellipse (0.8cm and 0.6cm);
\draw (6.9,5.1) ellipse (0.3cm and 0.6cm);

\draw (8.3,7.0) node [draw=none] (s4) {$U$};
\draw (5,4.9) node [draw=none] (s5) {$S^d$};

\draw (3.1,5.8) node [draw=none] (s0) {$A_{k}(\gamea)$};
\draw (7.8,4.9) node [draw=none] (s1) {$A_{1-k}(\gamea)$};

\draw (5,8.0) node [draw=none] (s2) {$W_{k}(\gamea')$};
\draw (6.8,8.0) node [draw=none] (s3) {$W_{1-k}(\gamea')$};


\draw (6,7.7) -- (6,8.5) {};
\draw (6,5.7) -- (6,4.2) {};


\end{tikzpicture}
\end{center}
\caption{Illustration of winning cores and winning regions in $\gamea$.}
\label{fig:recurse}
\end{figure}

\subsection{Memoryless strategies}

In this subsection we will show that from winning core states, player $j$ has a \emph{memoryless strategy} $\sigma$ which ensures that the play begins with an infinite number of consecutive $j$-dominating sequences. In fact, we will show something even stronger, namely the following.

\begin{theorem}
\label{theo:memless}
Let $\gamea$ be a parity game and $j$ be a player. There is a memoryless strategy for player $j$ such that

\begin{itemize}
\item $c(\play(\gamea, s, \sigma)) \subseteq \Lambda_j$ for every $s \in A_j(\gamea)$

\item $c(\play(\gamea, s, \sigma)) \subseteq \Omega_j$ for every $s \in W_j(\gamea)$

\item $c(\play(\gamea, s, \sigma)) \cap \Lambda_{1-j} = \emptyset$ for every $s \not \in A_{1-j}(\gamea)$ 
 
\end{itemize}
 
\end{theorem}

That is, player $j$ has a memoryless strategy that ensures that the play begins with an infinite number of consecutive $j$-dominating sequences when the play starts in a winning core state. Moreover, it ensures that the play is winning when the play begins in a winning state for player $j$. Finally, it ensures that the play does not begin with an infinite number of $(1-j)$-dominating sequences when the play does not begin in a winning core state of player $1-j$.

In order to prove Theorem \ref{theo:memless} we will use a result from \cite{GZ05}. But first we need a few definitions. Let a \emph{preference relation} be a binary relation on infinite sequences of colors (from a finite set of colors) that is reflexive, transitive and total. Let $\sqsubseteq_0$ be a preference relation for player 0. Intuitively, for two infinite sequences $\alpha$ and $\alpha'$ of colors we write $\alpha \sqsubseteq_0 \alpha'$ when $\alpha'$ is at least is good for player 0 as $\alpha$. As we deal only with antagonistic games here, we assume that there is a corresponding preference relation $\sqsubseteq_1$ for player $1$ such that for all infinite sequences $\alpha, \alpha'$ of colors we have $\alpha \sqsubseteq_0 \alpha'$ if and only if $\alpha' \sqsubseteq_1 \alpha$. We write $\alpha \sqsubset_j \alpha'$ for player $j$ when $\alpha \sqsubseteq_j \alpha'$ and $\alpha' \not \sqsubseteq_j \alpha$.

An \emph{optimal strategy} $\sigma_j^*$ for player $j$ in a game $\gamea$ with preference relation $\sqsubseteq_j$ is a strategy such that for every state $s$, all strategies $\sigma_j$ and $\sigma_{1-j}$ of player $j$ and $1-j$ respectively the unique plays $\rho^* \in \play(\gamea, s, \sigma_j^*) \cap \play(\gamea, s, \sigma_{1-j})$ and $\rho \in \play(\gamea, s, \sigma_j) \cap \play(\gamea, s, \sigma_{1-j})$ satisfy $\rho \sqsubseteq_j \rho^*$.

We now define a total order $\trianglelefteq_j$ for player $j$ with corresponding strict order $\vartriangleleft_j$ on $\{\Lambda_{j}, \Lambda_{1-j}, \Omega_j \setminus \Lambda_j, \Omega_{1-j} \setminus \Lambda_{1-j} \}$ by
$$\Lambda_{1-j} \vartriangleleft_j \Omega_{1-j} \setminus \Lambda_{1-j} \vartriangleleft_j \Omega_j \setminus \Lambda_j \vartriangleleft_j \Lambda_j$$
Note that an infinite sequence $\alpha$ of colors in a parity game belongs to exactly one of the four sets above. We write $\kappa(\alpha)$ for the set that $\alpha$ belongs to. For instance, $\kappa(\alpha) = \Omega_0 \setminus \Lambda_0$ for the infinite sequence $\alpha = 2 3 2 2 2 2...$

Now, more specifically, let $\leq_j$ be a preference relation for player $j$ on infinite sequences $\alpha, \alpha'$ of colors induced by the order $\trianglelefteq_j$ as follows
$$\alpha \leq_j \alpha' \textup{ if and only if } \kappa(\alpha) \trianglelefteq_j \kappa(\alpha')$$
As a special case of Proposition 7 in \cite{GZ05} we have the following.

\begin{proposition}
\label{prop:solitaire}
 Let player 0 have preference relation $\leq_0$ and player 1 have preference relation $\leq_1$. If
 
 \begin{enumerate}
  \item every parity game $\gamea = (S, S_0, S_1, R, c)$ with $S = S_0$ has a memoryless optimal strategy for player 0 and
  
  \item every parity game $\gamea = (S, S_0, S_1, R, c)$ with $S = S_1$ has a memoryless optimal strategy for player 1
 \end{enumerate}
 then in every parity game $\gamea$ both player 0 and player 1 have memoryless optimal strategies.
 
\end{proposition}

That is, there exists memoryless optimal strategies in every game if and only if there exists memoryless optimal strategies in every game where one player controls all the states.

As the preference relations $\leq_j$ are defined symmetrically, Proposition \ref{prop:solitaire} tells us that if we can show that player 0 has a memoryless optimal strategy with preference relation $\leq_0$ in all parity games where player $0$ controls every state then Theorem \ref{theo:memless} follows. This is in fact the case.

\begin{proposition}
\label{prop:memoryless}
 Let $\gamea = (S, S_0, S_1, R, c)$ be a parity game with $S = S_0$. Then player 0 has a memoryless optimal strategy with preference relation $\leq_0$.
\end{proposition}

\subsection{A sequence that converges quickly to the winning core}

Let $A^i_j(\gamea)$ be the set of states from which player $j$ can ensure that the play begins with at least $i$ consecutive $j$-dominating sequences. First, note that this defines an infinite decreasing sequence
$$A^0_j(\gamea) \supseteq A^1_j(\gamea) \supseteq ...$$
of sets of states and that $A^i_j(\gamea) \supseteq A_j(\gamea)$ for all $i \geq 0$.

\begin{theorem}
\label{theo:n_jdoms}
$A^{n}_j(\gamea) = A_j(\gamea)$
\end{theorem}

\begin{proof} First note that $A^{n}_j(\gamea) \supseteq A_j(\gamea)$.

To show that $A^{n}_j(\gamea) \subseteq A_j(\gamea)$ suppose for contradiction that $s \in A^{n}_j(\gamea)$ and $s \not \in A_j(\gamea)$. By Theorem \ref{theo:memless} player $1-j$ has a memoryless strategy $\sigma$ in $\gamea$ such that $c(\play(\gamea, s, \sigma)) \cap \Lambda_j = \emptyset$.

Since $s \in A^{n}_j(\gamea)$ there exists a play $\rho \in \play(\gamea, s, \sigma)$ that begins with $n$ consecutive $j$-dominating sequences. Let $i_0 < ... < i_n$ be indices with $i_0 = 0$ such that $\rho_{i_k} \rho_{i_k+1} ... \rho_{i_{i+1}}$ is $j$-dominating for all $0 \leq k < n$. As there are $n+1$ indices and only $n$ different states in $\gamea$ there must exists two indices $u,v$ with $u < v$ such that $\rho_{i_u} = \rho_{i_v}$. Now, the play $\pi = \rho_0 ... \rho_{u} (\rho_{u+1} ... \rho_{v})^\omega$ belongs to $\play(\gamea, s, \sigma)$ as well since $\sigma$ is memoryless. This gives a contradiction since $c(\pi) \in \Lambda_j$ and $\play(\gamea, s, \sigma)$ contains no such play according to the definition of $\sigma$.
\end{proof}

This proposition implies that if there is a way to calculate $A^i_j(\gamea)$ from $A^0_j(\gamea), ..., A^{i-1}_j(\gamea)$ in polynomial time for a given $i$ then the winning core can be computed in polynomial time as the sequence converges after at most $n$ steps. This would also imply that parity games could be solved in polynomial time.

To illustrate why it is not necessarily easy to compute this in a simple way consider again the parity game in Figure \ref{fig:not_dominion} and the history $h = s_0 s_1 s_2 s_2 s_2 s_2$ which begins with the 5 consecutive $0$-dominating sequences $s_0 s_1$, $s_1 s_2$, $s_2 s_2$, $s_2 s_2$ and $s_2 s_2$. As player 1 controls every state of the game he might force the play after $h$ to continue with the suffix $s_3 s_0^\omega$. Now, the way we chopped up $h$ into $5$ consecutive $0$-dominating sequences cannot be extended such that the entire play $\rho = h \cdot s_3 s_0^\omega$ begins with an infinite number of consecutive $0$-dominating sequences as the color of $s_3$ is larger than all colors that appear later in $\rho$. However, if we pick the first $0$-dominating sequence to be $s_0 s_1 s_2 s_2 s_2 s_2 s_3$ then it is easy to see that $\rho$ begins with an infinite number of consecutive $0$-dominating sequences. Thus, during the play of a game we might not know how to chop up the play in a way which ensures that the play begins with an infinite number of consecutive $0$-dominating sequences when the play begins in a winning core state for player $0$. However, we know that it is possible for player $0$ to force that the play begins with an infinite number of consecutive $0$-dominating sequences.

\section{Complexity of winning core computation}
\label{sec:complexity_wc}
In this section we show how computation of winning cores can be used to solve parity games. Next, we provide a polynomial-time reduction of solving parity games to computing winning cores.

\subsection{Solving parity games by winning core computation}

By Proposition \ref{prop:core_subset} the winning core for player $j$ is a subset of the winning region. Thus, according to Proposition \ref{prop:remove_winning} we get the following corollary which forms the basis of a recursive algorithm for solving parity games by computing winning cores.

\begin{corollary}
 \label{cor:remove_core}
 
  Let $\gamea = (S, S_0, S_1, R, c)$ be a parity game, $A' = \attr_j(\gamea, A_j(\gamea))$ and $\gamea' = \gamea \restriction (S \setminus A')$. Then
 
 \begin{itemize}
  \item $W_j(\gamea) = A' \cup W_j(\gamea')$
  
  \item $W_{1-j}(\gamea) = W_{1-j}(\gamea')$
  
 \end{itemize}
\end{corollary}

Given an algorithm $\textsc{WinningCore}(\gamea, j)$ that computes the winning core $A_j(\gamea)$ for player $j$ in $\gamea$ we can compute winning regions in parity games using the algorithm in Figure \ref{fig:algo_parity}.

 \begin{figure}
 \centering
  \centering
  \algtext*{EndIf}
 \algtext*{EndFor}

\begin{algorithmic}
\State \textsc{ParityGameSolver}($\gamea$):


\State $A \gets \textsc{WinningCore}(\gamea, 0)$
\If{$A = \emptyset$} \Return $(\emptyset, S)$
\EndIf
  \State $A' = \attr_0(\gamea, A)$
  \State $(W_0,W_1) \gets \textsc{ParityGameSolver}(\gamea \restriction (S \setminus A'))$
  \State \Return $(A' \cup W_0, W_1)$

\end{algorithmic}
\caption{Solving parity games using winning core computation}
\label{fig:algo_parity}
\end{figure}

The algorithm first calculates the winning core for player 0. If it is empty then by Theorem \ref{theo:empty_iff} player 1 wins in all states. Otherwise, $A' = \attr_0(\gamea, A_0(\gamea))$ is winning for player 0 and further, the remaining winning states can be computed by a recursive call on $\gamea \restriction (S \setminus A')$ according to Corollary \ref{cor:remove_core}. Note that this game has a strictly smaller number of states than $\gamea$ as $A' \neq \emptyset$. Thus, the algorithm performs at most $n$ recursive calls. This implies that if winning cores can be computed in polynomial time then parity games can be solved in polynomial time.


\subsection{Reducing winning core computation to solving parity games}

We have seen how existence of a polynomial-time algorithm for computing winning cores would imply the existence of a polynomial-time algorithm for solving parity games. Here, we show the converse by a reduction from computing winning cores to solving parity games. 

We begin by introducing the notion of a \emph{product game} $\gamea_j^\dagger$ of a parity game $\gamea$ for player $j$.  Let $\gamea = (S, R, S_0, S_1, c)$ be a parity game with colors in $\{1,...,d\}$ and  $j \in \{0,1\}$ be a player. Construct from this a game $\gamea_j^\dagger = (S', S'_0, S'_1, R', c')$ such that $S' = S \times \{0,1,...,d\}$, $S'_j = S_j \times \{0,1,...,d\}$ for $j \in \{0,1\}$, $R' = \{((s,v),(s',v')) \in S' \times S' \mid (s,s') \in R \wedge v' = \max(v,c(s')) \}$. Finally, $c'(s,v) = c(s)$ if $v \equiv j \textup{ (mod } 2)$ and $c'(s,v) = v$ otherwise.

\mycut{
\begin{itemize}
 \item $S' = S \times \{0,1,...,d\}$
 
 \item $S'_0 = S_0 \times \{0,1,...,d\}$
 
 \item $S'_1 = S_1 \times \{0,1,...,d\}$
 
 \item $R' = \{((s,v),(s',v')) \in S' \times S' \mid (s,s') \in R \wedge v' = \max(v,c(s')) \}$
 
 \item $c'(s,v) = \left\{ \begin{array}{rl}
c(s) &\mbox{ if $v \equiv j \textup{ (mod } 2)$} \\
v &\mbox{ otherwise}
\end{array} \right. $
 
\end{itemize}
}
The idea is that the rules of $\gamea_j^\dagger$ when the play starts in $(s,0)$ are the same as in $\gamea$ when the play starts in $s$, but with two main differences. The first is that in $\gamea_j^\dagger$, the greatest color that has occured during the play (excluding the color of the initial state) is recorded in the state. The second is that the color of states in $\gamea_j^\dagger$ are as in $\gamea$ when the greatest color $e$ occuring so far in the play satisfies $e \equiv j \textup{ (mod } 2)$. Otherwise, the state is colored $e$.

We define a bijection $\gamma_\gamea: \path(\gamea) \times \{0,...,d\} \rightarrow \path(\gamea_j^\dagger)$ for $\rho = s_0 s_1 ... \in \path(\gamea)$ and $v \in \{0,...,d\}$ by
$$\gamma_\gamea(\rho, v) = (s_0,w_0) (s_1,w_1) ...$$
where $w_i = \max(v, \max_{0 < k \leq i} c(s_i))$. In particular, $w_0 = v$.

For a pair $(s,v) \in S \times \{0,...,d\}$ we define $\st(s,v) = s$ and $\val(s,v) = v$. This is extended to paths $\rho  = (s_0,v_0) (s_1,v_1) ...$ in $\gamea_j^\dagger$ such that the state sequence $\st(\rho) = s_0 s_1 ...$ and value sequence $\val(\rho) = v_0 v_1 ...$.

The following two lemmas show that $s \in S$ is in the winning core of player $j$ in $\gamea$ if and only if $(s,0)$ is a winning state for player $j$ in $\gamea_j^\dagger$. As $\gamea_j^\dagger$ has size polynomial in the size of $\gamea$ this gives a reduction from computing winning cores to computing winning regions.

\begin{lemma}
  Let $\gamea$ be a parity game and $s$ be a state. Then $s \in A_j(\gamea)$ implies $(s,0) \in W_j(\gamea_j^\dagger)$.
 
\end{lemma}

\begin{proof}
Let $s \in A_j(\gamea)$. Then player $j$ has a strategy $\sigma$ in $\gamea$ such that $c(\play(\gamea, s, \sigma)) \subseteq \Lambda_j$. Now, construct from this a strategy $\sigma'$ for player $j$ in $\gamea_j^\dagger$ defined by
$$\sigma'(h) = (\sigma(\st(h)), \max(v_\ell,c(\sigma(\st(h)))))$$
for every history $h = (s_0,v_0) ... (s_\ell, v_\ell)$ in $\gamea_j^\dagger$.

Consider an arbitrary play
$$\rho' = (s_0,v_0) (s_1,v_1) ... \in \play(\gamea_j^\dagger, (s,0), \sigma')$$
from $(s,0)$ compatible with $\sigma'$. By the definition of $\sigma'$ we have that $\rho = s_0 s_1 ... \in \play(\gamea, s, \sigma)$ and thus $c(\rho) \in \Lambda_j$ as for every $i \geq 0$ such that $(s_i,v_i) \in S_j$ we have $s_{i+1} = \sigma(s_0 ... s_i)$. 

Let $e$ be the largest color that occurs in $\rho$. As $c(\rho) \in \Lambda_j$ we have $e \equiv j \textup{ (mod } 2)$. Thus, there exists an $\ell$ such that $v_i = e$ for all $i \geq \ell$. This implies that $c'((s_i,v_i)) = c(s_i)$ for all $i \geq \ell$ by the definition of $\gamea_j^\dagger$. Thus, the sequence of colors occuring in $\rho'_{\geq \ell}$ is the same as in $\rho_{\geq \ell}$ and therefore $c(\rho') \in \Omega_j$ using Lemma \ref{lem:suffix_winning}. As $\rho'$ was chosen arbitrarily from $\play(\gamea_j^\dagger, (s,0), \sigma')$ we have $(s,0) \in W_j(\gamea_j^\dagger)$.
\end{proof}

\begin{lemma}
  Let $\gamea$ be a parity game and $s$ be a state. Then $(s,0) \in W_j(\gamea_j^\dagger)$ implies $s \in A_j(\gamea)$.
 
\end{lemma}

\begin{proof}

Suppose that $(s,0) \in W_j(\gamea_j^\dagger)$. Then player $j$ has a strategy $\sigma'$ in $\gamea_j^\dagger$ such that $c(\play(\gamea_j^\dagger, (s,0), \sigma')) \subseteq \Omega_j$. Define a strategy $\sigma$ of player $j$ for every history $h$ in $\gamea$ by
$$\sigma(h) = \st(\sigma'(\gamma_\gamea(h,0)))$$

Consider an arbitrary play $\rho = s_0 s_1 ... \in \play(\gamea, s, \sigma)$ from $s$ compatible with $\sigma$. By the definition of $\sigma$ and $\gamea_j^\dagger$ we have that
$$\rho'= \gamma_\gamea(\rho, 0)$$
belongs to $\play(\gamea_j^\dagger, (s,0), \sigma')$ and thus $c(\rho') \in \Omega_j$. This implies that the greatest color $e$ occuring in $\rho'$ satisfies $e \equiv j \textup{ (mod } 2)$ by the definition of $\gamea_j^\dagger$. Further, the greatest color $e'$ that occurs infinitely often in $\rho'$ also satisfies $e' \equiv j \textup{ (mod } 2)$. We have that $e$ and $e'$ are also the greatest color occuring and greatest color occuring infinitely often respectively in $\rho$. Using Proposition \ref{prop:cons_iff_jdom} this implies that $c(\rho) \in \Lambda_j$. As $\rho$ is an arbitrary play in $\play(\gamea, s, \sigma)$ we have $s \in A_j(\gamea)$.
 \end{proof}

This means that solving parity games can be done in polynomial time if and only if winning cores can be computed in polynomial time. We also have that computing winning cores is in $\np \cap \conp$ and $\up \cap \coup$ like parity games \cite{Jur98}.

\begin{theorem}
 Computing winning cores is in $\np \cap \conp$ and $\up \cap \coup$.
\end{theorem}

 This fact is important as it makes the search for a polynomial-time algorithm for computing winning cores a viable direction in the search for a polynomial-time algorithm for solving parity games. This had not been the case if computing winning cores was e.g. $\np$-hard (which is still possible, but only if $\np = \conp$).

\section{A polynomial-time approximation algorithm}
\label{sec:approx}
A natural approach for computing the winning core is to try to apply Theorem \ref{theo:n_jdoms} using an algorithm resembling the standard algorithm for solving B\" uchi games using repeated attractor computations \cite{Tho95}. The idea is to first compute the set of states from which player $j$ can ensure that the play begins with one $j$-dominating sequence, then use this to compute the set of states from which player $j$ can ensure that the play begins with two consecutive $j$-dominating sequences, then three consecutive $j$-dominating sequences and so on until convergence. However, this turns out not to be so simple to do efficiently. In this section we propose a polynomial-time algorithm using the intuition above, but which is only guaranteed to compute a subset of the winning core. However, as we will see, this algorithm turns out to be very fast in practice and solves many games completely. We will show to make it work using $O(d \cdot n^2 \cdot (n+m))$ time and $O(d+n+m)$ space.

\subsection{The basic algorithm}

For a parity game $\gamea$, a player $j$ and integer $i \geq 0$ we define sets $B^i_j(\gamea)$ as underapproximations of $A^i_j(\gamea)$ by $B^0_j(\gamea) = S$ and by letting $B^{i+1}_j(\gamea)$ be the set of states from which player $j$ can force the play to begin with a $j$-dominating sequence ending in $B^i_j(\gamea)$. More formally, for every $i \geq 0$ let
\\

$B_j^{i+1}(\gamea) = \{s \in B_j^i(\gamea) \mid \exists \textup{ a strategy } \sigma \textup{ for player } j .$
\\

$\forall \rho \in \play(\gamea, s, \sigma) . \exists k . \rho_{\leq k}$ is $j$-dominating and $\rho_k \in B^i_j(\gamea) \}$
\\
\\
Note that this sequence converges in at most $n$ steps since it is decreasing. Let the limit of this sequence be $B_j(\gamea)$.

\begin{proposition}
\label{prop:b_subset_a}
 $B_j(\gamea) \subseteq A_j(\gamea)$
\end{proposition}

\begin{remark}
Note that we do not always have $B_j(\gamea) = A_j(\gamea)$. For instance, this is not the case in the game in Figure \ref{fig:not_dominion} where $A_0(\gamea) = \{s_0,s_3\}$ but $B^1_0(\gamea) = \{s_0,s_1,s_3\}$, $B^2_0(\gamea) = \{s_0,s_3\}$, $B^3_0(\gamea) = \{s_3\}$ and $B^4_0(\gamea) = B_0(\gamea) = \emptyset$. The reason is that from $s_0$ player 0 cannot force the play to ever get back to the set $\{s_0,s_3\}$ as player 1 controls all states. It can be shown that $A^1_j(\gamea) = B^1_j(\gamea)$ always, but there are parity games where $A^2_j(\gamea) \neq B^2_j(\gamea)$.
\end{remark}

However, as we shall see later, this underapproximation of the winning core is very good as a tool to compute underapproximations of winning regions in parity games. And in practice, it is often good enough to compute the entire winning regions. Moreover, we will show that $B_j(\gamea)$ can be computed in polynomial time and linear space. To motivate the practicability we note that the underapproximation $B_j(\gamea)$ contains all fatal attractors for player $j$. It was shown in \cite{HKP13,HKP14} that just being able to compute fatal attractors is enough to solve a lot of games in practice. In Figure \ref{fig:more_than_fatal} it is an example that $B_j(\gamea)$ can contain even more states than just fatal attractors.

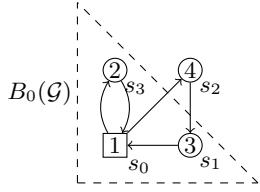
\begin{figure}
\begin{center}
\begin{tikzpicture}

\tikzstyle{every node}=[ellipse, draw=black,
                        inner sep=0pt, minimum width=9pt, minimum height=9pt]

\draw (5,3) node [rectangle, label = below right: $s_0$] (s0) {$1$};    
\draw (6,3) node [label = below right: $s_1$] (s1) {$3$};
\draw (6,4) node [label = below right: $s_2$] (s2) {$4$};
\draw (5,4) node [label = below right: $s_3$] (s3) {$2$};
\draw (4,3.7) node [draw=none] (s4) {$B_0(\gamea)$};

\path[->] (s0) edge (s2) {};
\path[->] (s0) edge [bend left] (s3) {};
\path[->] (s1) edge (s0) {};
\path[->] (s2) edge (s1) {};
\path[->] (s3) edge [bend left] (s0) {};

\path[-] (4.5,2.5) edge [dashed] (4.5,4.9);
\path[-] (4.5,2.5) edge [dashed] (6.9,2.5);
\path[-] (6.9,2.5) edge [dashed] (4.5,4.9);

\end{tikzpicture}
\end{center}

\caption{A parity game with no fatal attractors where $B_0(\gamea) = A_0(\gamea) = \{s_0,s_1,s_3\}$.}
\label{fig:more_than_fatal}
\end{figure}

Let $[1,d]_j = \{v \in \{1,...,d\} | v \equiv j \textup{ (mod } 2)\}$. We can now show the following proposition which provides us with a naive way to compute $B^{i+1}_j(\gamea)$ given that we know $B^i_j(\gamea)$. 

\begin{lemma}
\label{lem:prod_approx}
 Let $i \geq 0$ be an integer and $j$ be a player. Then $s \in B^{i+1}_j(\gamea)$ if and only if $(s,0) \in \attr_j(\gamea_j^\dagger, B^i_j(\gamea) \times [1,d]_j)$
\end{lemma}

Note that Lemma \ref{lem:prod_approx} makes us able to compute $B_j(\gamea)$ in time $O(d \cdot n \cdot (n+m))$ and space $O(d \cdot(n+m))$. This is because the sequence converges in at most $n$ steps and in each step we just have to compute the attractor set in $\gamea_j^\dagger$ which has $O(d \cdot n)$ states and $O(d \cdot m)$ transitions.

\subsection{Improving the complexity}

We will now show how to improve the space complexity to $O(d+n+m)$ while keeping the same time complexity. Indeed, we will show how to compute $B_j(\gamea)$ without actually having to construct $\gamea_j^\dagger$ explicitly. This makes a very large difference in practice, especially when the number of colors is large.

First, we need the reward order $\prec_j$ for player $j$ on colors which was introduced in \cite{VJ00}. It is defined by
$$v \prec_j u \Leftrightarrow (v < u \wedge u \equiv j \textup{ (mod 2)}) \vee (u < v \wedge v \equiv 1-j \textup{ (mod 2)})$$
We let it be defined for $0$ in this way as well. Intuitively, the preference order tells us which color player $j$ would rather like to see during a play. For instance, if $d$ is even then
$$d-1 \prec_0 d-3 \prec_0 ... \prec_0 1 \prec_0 0 \prec_0 2 \prec_0 ... \prec_0 d-2 \prec_0 d$$
We can now show that the attractor set needed to compute $B^{i+1}_j(\gamea)$ in $\gamea_j^\dagger$ is upward-closed in the following sense.

\begin{lemma}
\label{lem:upward_closed}
 Let $s \in S$, $i \geq 0$ and $k \geq 0$ be such that $(s,k) \in \attr_j(\gamea_j^\dagger, B^i_j(\gamea) \times [1,d]_j)$. Then for all $k' \succ_j k$ we have
 $$(s,k') \in \attr_j(\gamea_j^\dagger, B^i_j(\gamea) \times [1,d]_j) $$
\end{lemma}

\begin{proof}
For a strategy $\sigma$ in $\gamea_j^\dagger$ and $v \in \{0,...,d\}$ define a strategy $\sigma_v$ for every history $s_0 ... s_\ell$ in $\gamea$ by
$$\sigma_v(s_0...s_\ell) = \st(\sigma((s_0,v) ... (s_\ell, \max(v, \max_{0 < \ell' \leq \ell} c(s_{\ell'})))))$$

 Suppose $(s,k) \in \attr_j(\gamea_j^\dagger, B^i_j(\gamea) \times [1,d]_j)$. Then there exists a strategy $\sigma$ for player $j$ in $\gamea_j^\dagger$ such that for every $\rho \in \play(\gamea, (s,k), \sigma)$ there exists $q$ such that $\rho_q \in B^i_j(\gamea) \times [1,d]_j$. Let $k' \succ_j k$ and let $\sigma'$ be a strategy for player $j$ in $\gamea_j^\dagger$ defined by
 \\
 
 $\sigma'((s_0,v_0) ... (s_\ell, v_\ell))  = (\sigma_{k}(s_0...s_\ell), \max(v_\ell, c(\sigma_{k}(s_0...s_\ell)))) $
 \\
 \\
 for every history $h = (s_0,v_0) ... (s_\ell,v_\ell)$ with $s_0 = s$ and $v_0 = k'$. 
 
 Now, consider a given play
 $$\rho' = (s_0,v_0) (s_1,v_1) ... \in \play(\gamea_j^\dagger, (s,k'), \sigma')$$
 By the definition of $\sigma'$ we have $s_0 s_1 ... \in \play(\gamea, s, \sigma_k)$. Further, the sequence
 $$\rho = (s_0,w_0) (s_1,w_1) ...$$
 where $w_0 = k$ and $w_{\ell+1} = \max(w_\ell, c(s_{\ell+1}))$ for all $\ell \geq 0$ belongs to $\play(\gamea, (s,k), \sigma)$. Thus, there exists $q$ such that $\rho_q \in B^i_j(\gamea) \times [1,d]_j$. This means that either

 \begin{enumerate}
  \item $w_0 \equiv j \textup{ (mod 2)}$ and $w_0 = w_q$ or
  
  \item $w_q = \max(c(s_\ell))_{1 \leq \ell \leq q} > w_0$
 \end{enumerate}

 In the first case we have $v_0 > w_0$ and $v_0 \equiv j \textup{ (mod 2)}$ since $w_0 \equiv j \textup{ (mod 2)}$ and $v_0 \succ_j w_0$ which implies $v_q = v_0$. Thus, in this case $(s_q,v_q) \in B^i_j(\gamea) \times [1,d]_j$.
 
 In the second case, if $v_0 \leq w_0$ then $(s_q,v_q) = (s_q,w_q) \in B^i_j(\gamea) \times [1,d]_j$ immediately. On the other hand, suppose $v_0 > w_0$. Since $v_0 \succ_j w_0$ this implies $v_0 \equiv j \textup{ (mod 2)}$ by the definition of $\succ_j$. Since either $v_q = v_0$ or $v_q = \max(c(s_\ell))_{1 \leq \ell \leq q} = w_q$ this implies $(s_q,v_q) \in B^i_j(\gamea) \times [1,d]_j$.
\end{proof}

We can use Lemma \ref{lem:upward_closed} to compute $\attr_j(\gamea_j^\dagger, B^i_j(\gamea) \times [1,d]_j)$ as follows. In each step of the attractor computation we store for each state $s \in S$ the $\prec_j$-smallest value $k$ such that $(s,k)$ belongs to the part of the attractor set computed so far. Thus, in each step we store an $n$-dimensional vector $\textbf{k} = (k_0,...,k_{n-1})$ of these values, one for each state. In each step of the attractor computation we compute the $\prec_j$-smallest values $\textbf{k'}$ in the next step of the attractor computation based on $\textbf{k}$. Using a technique similar to the way the standard attractor set can be computed in time $O(n+m)$ \cite{dAHK98} the computation of $\attr_j(\gamea_j^\dagger, B^i_j(\gamea) \times [1,d]_j)$ can be done in this fashion in time $O(d \cdot (n+m))$. Thus, $B_j(\gamea)$ can be computed using $O(n+m+d)$ space and $O(d \cdot n \cdot (n+m))$ time.

\subsection{Partially solving parity games}

Using a similar idea as in the algorithm from Section \ref{sec:complexity_wc} we present an algorithm for solving parity games partially which relies on the underapproximation $B_j(\gamea)$. It can be seen in Figure \ref{fig:algo_approx}.

 \begin{figure}
 \centering
  \centering
  \algtext*{EndIf}
 \algtext*{EndFor}

\begin{algorithmic}
\State \textsc{PartialSolver}($\gamea$):


\State $A \gets B_0(\gamea)$
\If{$A \neq \emptyset$}
  \State $A' = \attr_0(\gamea, A)$
  \State $(W_0,W_1) \gets \textsc{PartialSolver}(\gamea \restriction (S \setminus A'))$
  \State \Return $(A' \cup W_0, W_1)$
\EndIf
\State $A \gets B_1(\gamea)$
\If{$A \neq \emptyset$}
  \State $A' = \attr_1(\gamea, A)$
  \State $(W_0,W_1) \gets \textsc{PartialSolver}(\gamea \restriction (S \setminus A'))$
  \State \Return $(W_0, A' \cup W_1)$
\EndIf

\State \Return $(\emptyset, \emptyset)$

\end{algorithmic}
\caption{A partial solver for parity games based on winning cores}
\label{fig:algo_approx}
\end{figure}

This algorithm uses the procedure outlined in the previous subsection for computing the underapproximation $B_j(\gamea)$. It is guaranteed to return underapproximations of the winning regions according to Proposition \ref{prop:b_subset_a}. Further, as each call to the algorithm makes at most one recursive call to a game with fewer states there are at most $O(n)$ recursive calls in total. Thus, the algorithm for partially solving parity games runs in time $O(d \cdot n^2 \cdot (n+m))$. It can be implemented to use $O(n+m+d)$ space.

\subsection{Quality of approximation}

For approximation algorithms a widely used notion is that of approximation ratio (see e.g. \cite{WS11}) which is used to give guarantees on the value of an approximation.

A meaningful way to define approximation ratio in parity games is to say that an algorithm is an $\alpha$-approximation algorithm for $0 < \alpha \leq 1$ if the algorithm always decides the winning player of at least $\lceil \alpha\cdot n \rceil$ states where $n$ is the number of states in the game. The problem with this, however, is that if there exists a polynomial-time $\alpha$-approximation algorithm for some $0 < \alpha \leq 1$ then this algorithm can be used to solve parity games completely in polynomial time. Indeed, one could run such an algorithm and remove the attractor sets of the winning states it finds. Then, run the approximation algorithm on the remaining game and continue in the same fashion until the entire winning regions are computed.

This tells us that it will probably be hard to show that there exists a polynomial-time $\alpha$-approximation algorithm as this would show solvability of parity games in polynomial time. In particular, our partial solver is not an $\alpha$-approximation algorithm.

A game that the partial solver cannot solve is the one in Figure \ref{fig:not_dominion}. The reason is that from every state player 1 can force the play to leave as well as stay outside of the winning core for player 0. This simple example implies that the algorithm is not guaranteed to solve games completely on standard subclasses of games investigated in the litterature such as games with bounded tree-width \cite{Obd03}, bounded DAG-width \cite{BDHK06} and other games with restrictions on the game graph \cite{DKT12}. Though, the algorithm always solves B\"uchi games completely and it does so in time $O(nm)$.

Despite the lack of theoretical guarantees we will show that the algorithm performs remarkably well in practice, with respect to solving games completely and with respect to running time.

\section{Experimental results}

We present experimental results for the improved version of the winning core approximation algorithm presented in Section \ref{sec:approx}, it is called the \texttt{WC} algorithm for the remainder of this section.

The experimental results are both performed to investigate how often the algorithm solves games completely and to investigate the running-time of the algorithm in practice compared to existing parity game solvers. The algorithm has been implemented in OCaml on top of the \textsc{PgSolver} framework \cite{FL09}.

We both compare with results for state-of-the-art complete solvers implemented in the \textsc{PgSolver} framework, namely

\begin{itemize}
 \item \texttt{Zie}: Zielonkas algorithm \cite{Zie98}
 
 \item \texttt{DD}: Dominion decomposition algorithm \cite{JPZ06}
 
 \item \texttt{SI}: Strategy improvement algorithm \cite{VJ00}
 
 \item \texttt{SPM}: Small progress measures algorithm \cite{Jur00}
 
 \item \texttt{BS}: Big step algorithm \cite{Sch07}
\end{itemize}
and with the partial solver \texttt{psolB} from \cite{HKP13,HKP14} that is based on fatal attractor computation. The experiments with the \texttt{WC} algorithm and the other solvers from the \textsc{PgSolver} framework have been performed on a machine with an Intel$^\circledR$ Core$^{TM}$ i7-4600M CPU with 4 2.90GHz processors and 15.6GiB memory.  All optimizations of the \textsc{PgSolver} framework were disabled in all experiments. The \texttt{WC} algorithm uses the same basic data structures as the other solvers from the \textsc{PgSolver} framework. All results of the partial solver \texttt{psolB} are taken from \cite{HKP14}. Thus, one should be careful about these results as it was implemented in Scala and experiments were run on a different machine.

\subsection{Benchmark games}

Experiments have been performed on benchmark games from the classes \texttt{Clique Games}, \texttt{Ladder Games}, \texttt{Jurdzinski Games}, \texttt{Recursive Ladder Games}, \texttt{Model Checker Ladder Games}, \texttt{Towers of Hanoi}, \texttt{Elevator Verification} and \texttt{Language Inclusion} of the \textsc{PgSolver} framework. \texttt{WC} solved all these benchmark games completely. As reported in \cite{HKP14} the \texttt{psolB} partial solver solves all games completely except for the \texttt{Elevator Verification} games.

Comparison of running time of the complete solvers and \texttt{WC} can be seen in Figure \ref{fig:experiments_benchmark} for selected benchmarks\footnote{For the recursive ladder games some solvers were much better with an odd input parameter and some were much better with an even input parameter. Thus, for each input $k$ in the data set, the running time on both input $k$ and $k+1$ was measured and the worst result is displayed in the plot.}. It can be seen that in the experiments \texttt{WC} never performs much worse than the best state-of-the-art complete solvers and in some cases it vastly outperforms the complete solvers. This is also the case for the results not shown here. Thus, it seems to be very robust compared to the best complete solvers each of which have games on which they perform poorly compared to the rest.

In Table \ref{tab:benchmark_results} we compare \texttt{WC} to the best existing partial solver \texttt{psolB} \cite{HKP14} with respect to the size of benchmark games solvable within 20 minutes. \texttt{WC} vastly outperforms \texttt{psolB} in all cases considered solving games with between 1.65 and 421 times as many states in 20 minutes depending on the benchmark.

\begin{figure}
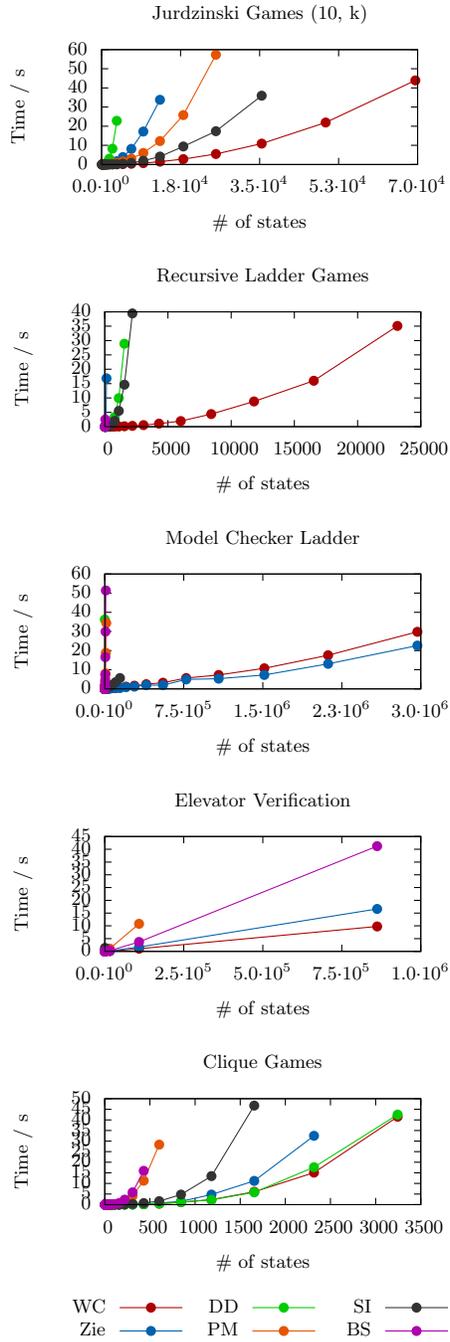

\centering
\scalebox{.8}{
 \input{jurdm=10.tex}
  }
  \scalebox{.8}{
 \input{recladder.tex}}
 \scalebox{.8}{
 \input{mcladder.tex}}
 \scalebox{.8}{
 \input{elevator.tex}}
\scalebox{.8}{
 \input{clique.tex}}

 \caption{Performance for benchmark games.}
 \label{fig:experiments_benchmark}
\end{figure}

\begin{table}
\begin{center}
 \begin{tabular}{l | r r | r r }
 Game & \multicolumn{2}{c}{\texttt{psolB}} & \multicolumn{2}{c}{\texttt{WC}}\\
 \hline
 & $k$ & $|S|$ & $k$ & $|S|$ \\
 \hline
  Clique & $5.232$ & $5.232$ & $8.979$ & $8.979$ \\
  Ladder & $1.596.624$ & $3.193.248$ & $7.308.357$ & $14.616.714$ \\
  Jur(10,k) & $2.890$ & $121.380$ & $4.784$ & $200.928$ \\
  Jur(k,10) & $4.380$ & $175.220$ & $96.881$ & $3.875.260$ \\
  Jur(k,k) & $200$ & $160.400$ & $635$ & $1.614.170$ \\
  RecLad & $2.064$ & $10.320$ & $14.008$ & $70.040$ \\
  MCLad & $12.288$ & $36.865$ & $5.178.332$ & $15.534.997$ \\
  Hanoi & $10$ & $236.196$ & $13$ & $6.377.292$ \\
  Elev & $6$ & $108.336$ & $8$ & $7.744.224$ \\
 \end{tabular}

\end{center}
\caption{The table shows the maximum input parameter $k$ and number $|S|$ of states for which the solvers terminated within 20 minutes. The numbers for \texttt{psolB} are from \cite{HKP14}.}
\label{tab:benchmark_results}

\end{table}

\subsection{Random games}

Although random games are not necessarily good representatives for real-life instances of parity games they can give us some indication of the quality of a partial solver. In order to compare with the results in \cite{HKP14} we have used the same program with the same parameters for generating random games. 

The games are generated using the \texttt{randomgame} function in the \textsc{PgSolver} framework which takes as input $n,d,\ell,u$ where $n$ is the number of states, the color of each state is chosen uniformly at random in $[1,d]$ and each node has a number of successors chosen uniformly at random in $[\ell,u]$ without any self-loops. 16 different configuration settings were chosen as in \cite{HKP14}: $n = 500$ for all parameter settings, $d \in \{5,50,250,500\}$ and $(\ell,u) \in \{(1,5),(5,10),(1,100),(50,250)\}$. 100.000 games were solved for each configuration and the results are shown in Table \ref{tab:random_results}.

\begin{table}
\begin{center}
 \begin{tabular}{c c r r r }
 $d$ & $(l,u)$ & {$\begin{array}{c}\textup{\# n.c.s.}\\\texttt{psolB}\end{array}$} & {$\begin{array}{c}\textup{\# n.c.s.}\\\texttt{lift(psolB)}\end{array}$} & {$\begin{array}{c}\textup{\# n.c.s.}\\\texttt{WC}\end{array}$}\\
 \hline
  5 & $(1,5)$ & 1275 & 233 & 258\\
  50 & $(1,5)$ & 1030 & 43 & 9\\
  250 & $(1,5)$ & 1138 & 36 & 16\\
  500 & $(1,5)$ & 1086 & 35 & 12 \\
  \hline
  5 & $(5,10)$ & 0 & 0 & 0\\
  50 & $(5,10)$ & 1 & 0 & 0\\
  250 & $(5,10)$ & 2 & 0 & 0\\
  500 & $(5,10)$ & 2 & 0 & 0 \\
 \end{tabular}

\end{center}
\caption{Table shows the number random of games out of 100.000 that were not completely solved (\# n.c.s.) by the solvers for each configuration. For all configurations with $(\ell,u) \in \{(1,100),(50,250)\}$ all solvers solved all games completely.}
\label{tab:random_results}

\end{table}

It can be seen that the algorithm only failed to solve 295 of the 1.600.000 games completely and thus solved 99.98\% of the games completely. For the 295 games that were not solved completely by the winning core algorithm it still found the winning player for 56\% of the states on average. Also note that the algorithm only failed to solve games with a very low out-degree as was the case for \texttt{psolB}. For the dense games it solved all the games completely.

Compared to \texttt{psolB} the \texttt{WC} solver does very well. The partial solver \texttt{lift(X)} is a generic solver from \cite{HKP14} which uses a partial solver X to improve in cases where X does not solve the complete game. It also runs in polynomial time but gives a very large overhead in practice as it potentially calls the solver X a quadratic number of times in the number of transitions of a game. Even compared with this generic optimization of \texttt{psolB} the \texttt{WC} solver does well with respect to solving games completely.

\section{Concluding remarks}

We have introduced winning cores and motivated their importance by showing a number of interesting properties about them. In particular, they are interesting to investigate due to the fact that they are not necessarily dominions and because emptiness of the winning core of a player is equivalent to emptiness of the winning region of the player. Further, we have provided a new algorithm for solving parity games approximately which increases the size of parity games that can be solved in practice siginificantly compared to existing techniques. 

I want to thank Michael Reichhardt Hansen and Valentin Goranko for helpful comments and discussions.

\bibliographystyle{plain}
\bibliography{biblio}

\appendix
\section{Proof of Proposition \ref{prop:winning_iff_jdoms}}

\begin{proof}[Proof of Proposition \ref{prop:winning_iff_jdoms}]

 $(\Rightarrow)$ Let $\rho$ be a play such that $c(\rho) \in \Omega_j$. Let $e \equiv j \textup{ (mod 2})$ be the greatest color occuring infinitely often in $\rho$. Since there are only a finite number of different colors, there can only be a finite number of indices $i$ such that $c(\rho_i) > e$. Let $\ell$ be the largest such index. Further, let $\ell < i_0 < i_1 < ...$ be an infinite sequence of indices such that $c(\rho_{i_k}) = e$ for all $k \geq 0$. Such a sequence exists since $e$ occurs infinitely often in $\rho$. Now, $\rho_{\geq i_0}$ begins with an infinite number of consecutive $j$-dominating sequences, namely the sequences $\pi_k = \rho_{i_k} \rho_{i_{k}+1} ... \rho_{i_{k+1}}$ for $k \geq 0$.

 $(\Leftarrow)$ Let $\rho$ be a play with a suffix $\rho_{\geq i_0}$ that begins with an infinite number of consecutive $j$-dominating sequences. Let $i_0 < i_1 < ...$ be an infinite sequence of indices such that $\pi_{\ell} = \rho_{i_\ell} \rho_{i_\ell + 1} ... \rho_{i_{\ell+1}}$ is $j$-dominating for every $\ell \geq 0$.
 
 Now, suppose for contradiction that $c(\rho) \in \Omega_{1-j}$. Then the greatest color $e$ that occurs infinitely often in $\rho$ satisfies $e \equiv 1-j \textup{ (mod } 2)$. Now, $e$ is the color of a non-initial state in $\pi_\ell$ for an infinite number of indices $\ell$. Since every such $\pi_\ell$ is $j$-dominating there is an infinite number of states in $\rho$ with a color $e' > e$ such that $e' \equiv j \textup{ (mod } 2)$. Since there are only finitely many different colors, there exists a particular color $e'' > e$ which is the color of infinitely many states in $\rho$. This gives a contradiction since $e$ was chosen as the greatest color that occurs infinitely often in $\rho$. This implies that $c(\rho) \in \Omega_j$.
\end{proof}

\section{Proofs of Proposition \ref{prop:core_1_k} and \ref{prop:core_k}}

First we need a general lemma about infinite games.

\begin{lemma}
 \label{lem:jclosed}
 
 Let $T \subseteq S$ be $j$-closed and $\sigma'$ be a strategy for player $j$ in $\gamea \restriction T$. Then there exists a strategy $\sigma$ for player $j$ in $\gamea$ such that
 $$\play(\gamea, s, \sigma) = \play(\gamea \restriction T, s, \sigma')$$
 for all $s \in T$. Moreover, if $\sigma'$ is memoryless then $\sigma$ can be chosen to be memoryless as well.
\end{lemma}

\begin{proof}
 Define $\sigma$ by $\sigma(h) = \sigma'(h)$ for every history $h$ that only contains states in $T$ and arbitrarily for all other histories.
 
 First we have that every play $\rho = s_0 s_1 ... \in \play(\gamea \restriction T, s, \sigma')$ where $s \in T$ belongs to $\play(\gamea, s, \sigma)$ as well since $\sigma(s_0 ... s_\ell) = \sigma'(s_0 ... s_\ell) = s_{\ell+1}$ for every prefix $s_0 ... s_\ell$ of $\rho$ such that $s_\ell \in S_j$ and $(s_\ell, s_{\ell+1}) \in R$ whenever $s_\ell \in S_{1-j}$.
 
 On the other hand, for every play $\rho = s_0 s_1 ... \in \play(\gamea, s, \sigma)$ where $s \in T$ we have that $s_i \in T$ for every $i \geq 0$. This can be shown by induction as follows. For the base case we have that $s_0 \in T$. For the induction step we have that whenever $s_i \in S_{1-j} \cap T$ then there exists no $t \in S \setminus T$ with $(s_i,t) \in R$ since $T$ is $j$-closed. Further, whenever $s_i \in S_j \cap T$ and $s_0,...,s_i \in T$ then $\sigma(s_0 ... s_i) = \sigma'(s_0 .., s_i) \in T$. This also implies that $\rho \in \play(\gamea \restriction T, s, \sigma')$. Thus, $\play(\gamea, s, \sigma) = \play(\gamea \restriction T, s, \sigma')$ for $s \in T$.
 
 Note that if $\sigma'$ is memoryless then $\sigma$ can be chosen to be memoryless as well.
\end{proof}

\begin{corollary}
 \label{cor:winjclosed}
 
 If $T \subseteq S$ is $j$-closed in $\gamea$ then $W_j(\gamea \restriction T) \subseteq W_j(\gamea)$.
\end{corollary}

We are now read to prove Proposition \ref{prop:core_1_k} and \ref{prop:core_k}.

\begin{proof}[Proof of Proposition \ref{prop:core_1_k}]
 
Suppose first that $s \in A_{1-k}(\gamea')$. Then there exists a strategy $\sigma'$ for player $1-k$ in $\gamea'$ such that every play $\rho \in \play(\gamea', s, \sigma')$ begins with an infinite number of consecutive $(1.k)$-dominating sequences. Let $\sigma$ be a strategy in $\gamea$ for player $1-k$ defined by $\sigma(h) = \sigma'(h)$ for histories $h$ that only contain states from $S'$. Let $\sigma$ be defined arbitrarily for all other histories. We now have that $\play(\gamea, s, \sigma) = \play(\gamea', s, \sigma')$ as player $k$ does not control any state in $S'$ with a transition to $U$ and $\sigma$ only prescribes taking transitions that make the play stay in $S'$ if no state outside $S'$ is reached. This implies that $s \in A_{1-k}(\gamea)$.

Suppose on the other hand that $s \in A_{1-k}(\gamea)$. Then there exists a strategy $\sigma$ for player $1-k$ such that every play $\rho \in \play(\gamea, s, \sigma')$ begins with an infinite number of consecutive $(1-k)$-dominating sequences. As $d \equiv k \textup{ (mod 2)}$ it follows from Proposition \ref{prop:cons_iff_jdom} that no state in a play $\rho \in \play(\gamea, s, \sigma')$ is contained in $U$. Thus, we can define a strategy $\sigma'$ in $\gamea'$ by $\sigma'(h) = \sigma(h)$ for every history with initial state $s$ and obtain $\play(\gamea, s, \sigma) = \play(\gamea', s, \sigma')$. This implies that $s \in A_k(\gamea')$.
\end{proof}

\begin{proof}[Proof of Proposition \ref{prop:core_k}]
 
Suppose first that $s \in A_{k}(\gamea'')$. Then there exists a strategy $\sigma''$ for player $k$ in $\gamea''$ such that every play $\rho \in \play(\gamea'', s, \sigma'')$ begins with an infinite number of consecutive $k$-dominating sequences. Let $\sigma$ be a strategy in $\gamea$ for player $k$ defined by $\sigma(h) = \sigma''(h)$ for histories $h$ that only contain states from $S''$. Let $\sigma$ be defined arbitrarily for all other histories. We now have that $\play(\gamea, s, \sigma) = \play(\gamea'', s, \sigma'')$ as player $1-k$ does not control any state in $S''$ with a transition to $S \setminus S''$ and $\sigma$ only prescribes taking transitions that make the play stay in $S''$ if no state outside $S''$ is reached. This implies that $s \in A_{k}(\gamea)$.

On the other hand suppose that $s \in A_{k}(\gamea)$. Then there exists a strategy $\sigma$ for player $k$ in $\gamea$ such that every play $\rho \in \play(\gamea, s, \sigma)$ begins with an infinite number of consecutive $k$-dominating sequences. Note that no such play has a state contained in $\attr_{1-k}(\gamea, A_{1-k}(\gamea')) = \attr_{1-k}(\gamea, A_{1-k}(\gamea))$ because all such states are winning for player $1-k$. Therefore, it is possible to define a strategy $\sigma''$ in $\gamea''$ for player $k$ by $\sigma''(h) = \sigma(h)$ for every history $h$ in $\gamea''$ with initial state $s$. Further, we get $\play(\gamea, s, \sigma) = \play(\gamea'', s, \sigma'')$ which implies $s \in A_{k}(\gamea'')$.
\end{proof}

\section{Proof of Proposition \ref{prop:memoryless}}

\begin{proof}[Proof of Proposition \ref{prop:memoryless}]

 First, we define the reward order $\prec_j$ for player $j$ on colors which was introduced in \cite{VJ00}. It is defined by
$$v \prec_j u \Leftrightarrow (v < u \wedge u \equiv j \textup{ (mod 2)}) \vee (u < v \wedge v \equiv 1-j \textup{ (mod 2)})$$

 Let $m(\rho)$ denote the largest non-initial color of a play $\rho$. For each state $s_0 \in W_0(\gamea)$ there exists a play $\rho = s_0 s_1 ...$ such that $\rho \in \Omega_0$. In particular, let $\val(s_0) = m(\rho)$ for a play $\rho \in \Omega_j$ from $s_0$ such that $m(\rho)$ is maximal with respect to the reward order $\preceq_0$.
 
 For states $s_0$ such that $\val(s_0) \succeq_0 0$ we define $\dist(s_0)$ to be the length of the shortest history $\rho = s_0 ... s_\ell$ from $s_0$ to a state $s_\ell$ with color $\val(s_0)$ such that for all $0 < i < \ell$ we have $c(s_i) < \val(s_0)$ and such that $\val(s_\ell) \succeq_0 \val(s_0) - 1$. Note that such a history must exist since $\val(s_0) \succeq_0 0$.
 
 For states $s_0$ such that $\val(s_0) \prec_0 0$ we define $\dist(s_0)$ to be the length of the shortest history $\rho = s_0 ... s_\ell$ from $s_0$ to a state $s_\ell$ with color $\val(s_0)$ such that for all $0 < i < \ell$ we have $c(s_i) < \val(s_0)$ and such that $\val(s_\ell) \succ_0 \val(s_0)$. Note that such a history must exist since $\val(s_0) \prec_0 0$ and $s_0 \in W_0(\gamea)$.
 
 For states $s$ with $\val(s) \succeq_0 0$ and $\dist(s) > 1$ there must exist a successor $t$ of $s$ such that $\val(t) = \val(s), c(t) < \val(s)$ and $\dist(t) = \dist(s) - 1$. Define $\nextstate(s) = t$ for an arbitrary such successor $t$.
 
 For states $s$ with $\val(s) \succeq_0 0$ and $\dist(s) = 1$ there must exist a successor $t$ of $s$ such that $c(t) = \val(s)$ and $\val(t) \succeq_0 \val(s) - 1$. Define $\nextstate(s) = t$ for an arbitrary such successor $t$.
 
 For states $s$ with $\val(s) \prec_0 0$ and $\dist(s) > 1$ there must exist a successor $t$ of $s$ such that $\val(t) = \val(s), c(t) < \val(s)$ and $\dist(t) = \dist(s) - 1$. Define $\nextstate(s) = t$ for an arbitrary such successor $t$.
 
 For states $s$ with $\val(s) \prec_0 0$ and $\dist(s) = 1$ there must exist a successor $t$ of $s$ such that $c(t) = \val(s)$ and $\val(t) \succ_0 \val(s)$. Define $\nextstate(s) = t$ for an arbitrary such successor $t$.
 
 Now, define a memoryless strategy $\sigma$ for player 0 by $\sigma(s) = \nextstate(s)$ for every $s \in W_0(\gamea)$. For states in $W_1(\gamea)$ there are no transitions to states in $W_0(\gamea)$ as $S = S_0$. As player 0 cannot win from a state in $W_1(\gamea)$ all he can hope to achieve is a play that is not in $\Lambda_1$ which can be obtained if and only if the largest color in the play is even. This is also known as the weak parity condition. As weak parity games are memoryless determined \cite{Cha08} player 0 can use a memoryless optimal strategy from the weak parity game obtained by restricting $\gamea$ to states in $W_1(\gamea)$. Let $\sigma$ play in this way from states in $W_1(\gamea)$.
 
 We will now show that $\sigma$ is a memoryless optimal strategy for player $0$ with preference relation $\leq_0$. As argued above this is already the case from states in $W_1(\gamea)$. Thus, we now focus on states in $W_0(\gamea)$.
 
 Let $\rho$ be the single play in $\play(\gamea, s_0, \sigma)$ for a state $s_0 \in W_0(\gamea)$. We will now show that $\rho \in \Omega_0$ and that $m(\rho) = \val(s_0)$. These two properties imply that $\sigma$ ensures that $\play(\gamea, s, \sigma) \subseteq \Omega_0$ for every $s \in W_0(\gamea)$ and that $\play(\gamea, s, \sigma) \subseteq \Lambda_0$ for every $s \in A_0(\gamea)$ as states $s$ in $A_0(\gamea)$ are exactly those states with $\val(s) \succeq_0 0$.
 
 Let $\rho = s_0 s_1 ...$ and let $i_0 < i_1 < ...$ be all indices such that $\dist(s_{i_k}) = 1$ for every $k \geq 0$. Notice that there are infinitely many such indices since $\dist(s_i) - 1 = \dist(s_{i+1})$ whenever $\dist(s_i) > 1$ and $\dist(s_i) \geq 1$ for all $i \geq 0$.
 
 As argued before, whenever $\dist(s_\ell) > 1$ we have $\val(s_\ell) = \val(s_{\ell + 1})$. When $\dist(s_\ell) = 1$ and $\val(s_\ell) \succeq_0 0$ then $\val(s_{\ell + 1}) \succeq_0 \val(s_\ell) - 1$ which implies that $\val(s_\ell) \geq \val(s_{\ell + 1})$. Finally, when $\dist(s_\ell) = 1$ and $\val(s_\ell) \prec_0 0$ then $\val(s_\ell + 1) \succ_0 \val(s_\ell)$ which implies that $\val(s_\ell) \geq \val(s_{\ell+1})$. Thus, we have
 $$\val(s_0) \geq \val(s_1) \geq \val(s_2) \geq ...$$
 For all $\ell$ such that $\dist(s_\ell) > 1$ we have $\val(s_{\ell+1}) = \val(s_\ell)$. However, when $\dist(s_\ell) = 1$ and $\val(s_\ell) \prec_0 0$ then $\val(s_\ell) > \val(s_{\ell+1})$. As there are infinitely many $\ell$ such that $\dist(s_\ell) = 1$ there can only be finitely many $\ell$ such that $\val(s_\ell) \prec_0 0$. Thus, there exists $q$ such that
 $$0 \preceq_0 \val(s_q) = \val(s_{q+1}) = \val(s_{q+2}) = ...$$
 Note that for all $\ell > 0$ we have $\val(\ell) \geq c(\ell)$. However, we have $\val(s_{i_k}) = c(s_{i_k + 1})$ for all $k \geq 0$. Let $p$ be the smallest index such that $i_p > q$. Then this implies that $c(s_{i_p + 1}) = \val(s_{i_p})$ is the largest color that occurs infinitely often in $\rho$. Thus, $\rho \in \Omega_0$. In addition, note that $c(s_{i_0 + 1}) = \val(s_{i_0}) = \val(s_0)$ is the largest non-initial color that occurs in $\rho$. Thus, $m(\rho) = \val(s_0)$. This concludes the proof that $\sigma$ is a memoryless optimal strategy.
 
\end{proof}

\section{Proof of Proposition \ref{prop:b_subset_a}}

\begin{proof}[Proof of Proposition \ref{prop:b_subset_a}]

 For every $s \in B_j(\gamea)$ there exists a strategy $\sigma_s$ for player $j$ such that for every $\rho \in \play(\gamea, s, \sigma_s)$ there exists $k$ such that $\rho_{\leq k}$ is $j$-dominating and $\rho_k \in B_j(\gamea)$.
 
 Now, define a strategy $\sigma$ for player $j$ in $\gamea$ on histories with initial state $s_0 \in B_j(\gamea)$ (and arbitrarily for all other initial states) as follows. Let $\sigma$ play like $\sigma_{s_0}$ until it reaches a state $s_1$ in a way such that the sequence of states from $s_0$ to $s_1$ is $j$-dominating and $s_1 \in B_j(\gamea)$. From this point on $\sigma$ plays like $\sigma_{s_1}$ would do if the play started in $s_1$ until a state $s_2$ is reached in a way such that the sequence of states from $s_1$ to $s_2$ is $j$-dominating and $s_2 \in B_j(\gamea)$. Let $\sigma$ prescribe continuing in this fashion indefinitely.
 
 With this definition we have that $\play(\gamea, s_0, \sigma) \subseteq \Lambda_j$ for every $s_0 \in B_j(\gamea)$ as every play must begin with an infinite number of consecutive $j$-dominating sequences. Thus, $s_0 \in B_j(\gamea)$ implies $s_0 \in A_j(\gamea)$.
\end{proof}

\section{Proof of Lemma \ref{lem:prod_approx}}

\begin{proof}[Proof of Lemma \ref{lem:prod_approx}]
$(\Rightarrow)$ Suppose that $s \in B^{i+1}_j(\gamea)$. Then there exists a strategy $\sigma$ for player $j$ in $\gamea$ such that for every $\rho = s_0 s_1 ... \in \play(\gamea, s, \sigma)$ there exists $k > 0$ such that $s_0 ... s_k$ is $j$-dominating and $s_k \in B^i_j(\gamea)$. Define a strategy $\sigma'$ in $\gamea_j^\dagger$ by
$$\sigma'(h) =( \sigma(\st(h)), \max(v_\ell,c(\sigma(\st(h))))) $$
for every history $h = (t_0,v_0) ... (t_\ell,v_\ell)$ in $\gamea_j^\dagger$. Let $\rho' = (s'_0,v'_0) (s'_1,v'_1) ... $ be an arbitrary play in $\play(\gamea_j^\dagger, (s,0), \sigma')$. Then $s'_0 s'_1 ... \in \play(\gamea, s, \sigma)$. Thus, there exists $k$ such that $s'_0 ... s'_k$ is $j$-dominating and $s'_k \in B^i_j(\gamea)$. This implies that $v'_{k} \equiv j \textup{ (mod 2)}$ and thus $(s'_k,v'_k) \in B^i_j(\gamea) \times [1,d]_j$. As $\rho'$ was arbitrarily chosen in $\play(\gamea_j^\dagger, (s,0), \sigma')$ this means that $(s,0) \in \attr_j(\gamea_j^\dagger, B^i_j(\gamea) \times [1,d]_j)$.

$(\Leftarrow)$ For the other direction suppose that $(s,0) \in \attr_j(\gamea_j^\dagger, B^i_j(\gamea) \times [1,d]_j)$. Then there exists a strategy $\sigma'$ for player $j$ in $\gamea_j^\dagger$ such that for every $\rho' = (s_0,v_0) (s_1,v_1) ... \in \play(\gamea_j^\dagger, (s,0), \sigma')$ there exists $k$ such that $\rho'_k \in B^i_j(\gamea) \times [1,d]_j$. Define a strategy $\sigma$ for player $j$ in $\gamea$ by
$$\sigma(h) = \st(\sigma'(\gamma_\gamea(h,0)))$$
for every history $h$ in $\gamea$. Let $\rho = s_0 s_1 ...$ be an arbitrary play in $\play(\gamea, s, \sigma)$. By the definition of $\sigma$ and $\gamea_j^\dagger$ we have that
$$\rho'= \gamma_\gamea(\rho, 0) $$
belongs to $\play(\gamea_j^\dagger, (s,0), \sigma')$. Thus, there exists $k$ such that $\rho'_k \in B^i_j(\gamea) \times [1,d]_j$. This implies that $s_0 ... s_k$ is $j$-dominating and that $s_k \in B^i_j(\gamea)$. As $\rho$ was chosen arbitrarily from $\play(\gamea, s, \sigma)$ this means that $s \in B^{i+1}_j(\gamea)$.
\end{proof}

\section{Additional Experimental Results}

Experimental results for the benchmarks that was left of the main part of the paper due to space restrictions can be seen in Figure \ref{fig:experiments_benchmark_2} and \ref{fig:experiments_benchmark_3}. It can be seen that the \texttt{WC} algorithm is the best performing solver in many cases here as well and that it is close to the performance of the best complete solver in the remaining cases.

\begin{figure}
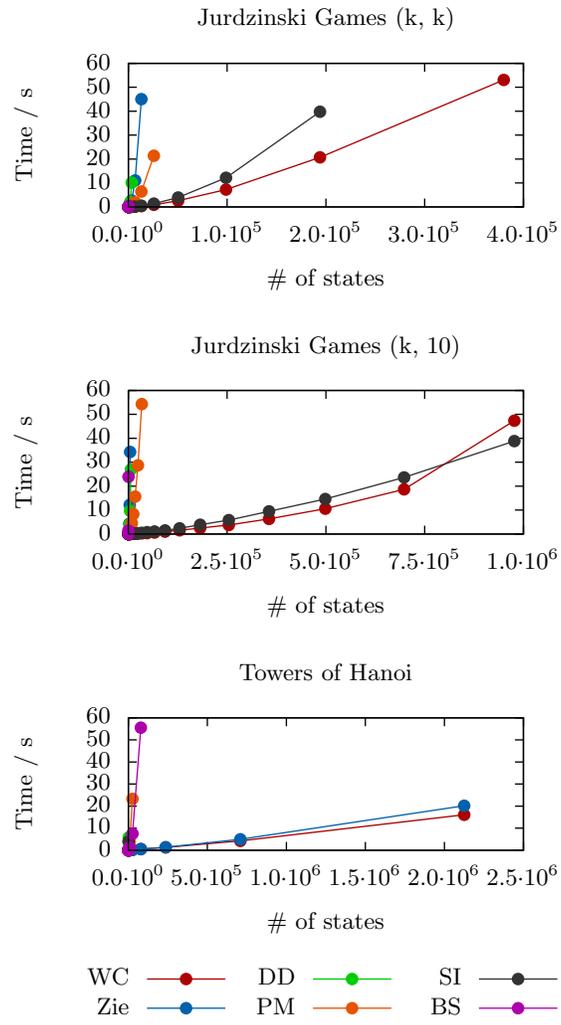

\centering
 \input{jurdnn.tex}
 
 \input{jurdn=10.tex}
 
 \input{hanoi.tex}

 \caption{Performance for benchmark games.}
 \label{fig:experiments_benchmark_2}
\end{figure}

\begin{figure}
\centering
 \input{langn=10.tex}
 
 \input{langm=10.tex}

 \input{ladder.tex}
 \caption{Performance for benchmark games.}
 \label{fig:experiments_benchmark_3}
\end{figure}

\end{document}